\documentclass[journal]{IEEEtran}
\ifCLASSINFOpdf
\else
\fi
\usepackage{graphicx}
\usepackage{amsmath}
\usepackage{amsfonts}
\usepackage{amssymb}
\usepackage{subcaption}
\usepackage{amsthm}
\newtheorem{proposition}{Proposition}

\theoremstyle{plain}

\newtheorem{lemma}{Lemma}

\theoremstyle{definition}
\newtheorem{assumption}{Assumption}

\theoremstyle{remark}
\newtheorem{remark}{Remark}

\usepackage{amsmath}

\begin{document}
%
\title{Adaptive Pilot Selection for Unified Semantic Communication and Semantic Sensing in ISAC}
%
%
%

\author{Muhammad Abubakar Rashid,
        Muhammad Hannan Akram,\\
        Haejoon Jung,~\IEEEmembership{Senior Member,~IEEE}
        and Syed Ali Hassan,~\IEEEmembership{Senior Member,~IEEE}
\thanks{Muhammad Abubakar Rashid, Muhammad Hannan Akram and Syed Ali Hassan are with the School of Electrical Engineering and Computer Science, National University of Sciences and Technology, Islamabad, Pakistan (e-mails: \{mrashid.bsds23seecs, makram.bsds23seecs, ali.hassan\}@seecs.edu.pk)}
\thanks{Haejoon Jung is with the Department of Electronic Engineering, Kyung Hee University, Yongin 17104, Republic of Korea (e-mail: haejoonjung@khu.ac.kr).}%
}
\maketitle

\begin{abstract}
Semantic communication (SemCom) and integrated sensing and communication (ISAC) are promising technologies for future 6G wireless networks. Existing studies have applied semantic technology to either the communication module or the sensing module of ISAC. In
this work, we propose \textbf{SemISAC}, which performs both SemCom and semantic sensing within a single dual-function waveform. SemISAC uses a joint semantic encoder that extracts task-specific information for both communication and sensing. We evaluate SemISAC in a vehicular scenario in which vehicles share pixel-wise segmentation of the road environment and, through sensing, classify surrounding objects and estimate their ranges. On the transmitter side, a deep learning encoder converts the input road-scene image into semantic symbols and places them on the data cells of an OFDM grid, while the remaining cells serve as pilots for channel state information estimation and sensing. The pilot configuration is adaptively optimized based on the channel conditions to balance communication and sensing requirements. At the receiver, a deep learning model reconstructs the segmentation from the received waveform, while the transmitting vehicle captures the reflected waveforms from surrounding objects and uses task-specific deep learning decoders for target recognition and range estimation. Simulation results show that SemISAC achieves a segmentation accuracy close to that of the dedicated SemCom module while outperforming both
conventional and semantic baselines in target recognition and range estimation.
\end{abstract}


%
\IEEEpeerreviewmaketitle

\section{Introduction}

\IEEEPARstart{I}{ntegrated} sensing and communication (ISAC) combines ubiquitous connectivity and advanced sensing capabilities in a unified framework and is critical to autonomous vehicles, smart cities, non-terrestrial networks (NTNs), and industrial automation, thereby emerging as a key enabler for future 6G networks \cite{ISAC1}, \cite{ISAC2}. In addition to ISAC, semantic communication (SemCom) has also attracted significant interest as a fundamental component of 6G networks, as it shifts the design objective from bit transmission to meaning exchange, facilitating the introduction of high-level semantics in wireless networks \cite{semcom1}, \cite{semcom2}. When combined, these paradigms enable wireless systems to not only sense and communicate information but also understand and extract its underlying meaning.

Systems incorporating SemCom into ISAC systems have attracted significant attention in the recent literature, focusing primarily on conveying semantic features learned from the data using deep learning (DL) techniques. By doing so, and by learning the properties of the wireless channel, SemCom enables robust transmissions and reliable task execution. Moreover, the tradeoff between sensing and communication is also mitigated in SemCom ISAC systems, which maintain downstream task accuracy comparable to that of conventional systems while improving sensing performance. Studies such as \cite{semISAC1}, \cite{semISAC2}, \cite{semISAC3}, \cite{semISAC4}, \cite{semISAC5}, and \cite{semISAC6} explore the use of SemCom for different purposes in ISAC systems. Reference \cite{semISAC1} presents a framework that combines SemCom with multimodal radar and visual sensing in a unified ISAC system. A semantic-based channel state information (CSI) feedback scheme was proposed in \cite{semISAC2}, which replaces full channel data transmission with lightweight semantic database labels to reduce the feedback burden for autonomous-aerial-vehicle-assisted ISAC tasks. In \cite{semISAC3}, a framework that jointly optimizes semantic symbol quantization, modulation order, power allocation, and dual-function radar-communication beamforming is proposed for image-based SemCom. Reference \cite{semISAC4} proposes a semantic-importance-guided deep reinforcement learning framework for multimodal SemCom in UAV-based ISAC, while \cite{semISAC6} incorporates SemCom into UAV-RIS-assisted ISAC by optimizing the transmission of task-relevant semantic symbols under semantic accuracy and secrecy constraints. The work in \cite{semISAC5} focuses on security-aware semantic ISAC, utilizing paired adversarial residual networks to protect semantic information from eavesdropping.

In addition to the rapid evolution of SemCom, enabling the same level of intelligence in the sensing domain is also an active area of research. The idea of semantic sensing (SemSens) has its theoretical foundations in cognitive radar \cite{semsens1}, which established a closed feedback loop between the transmitter and receiver. Subsequently, many studies were conducted to adjust sensing strategies based on environmental knowledge, such as adaptive waveform design and predictive beamforming \cite{semsens5}. In \cite{semsens2}, the authors developed a sensing-assisted beamforming technique that employs particle filtering and particle swarm optimization (PSO) and utilizes multipath echoes for vehicle tracking and beamforming optimization. Similarly, an adaptive waveform design method was proposed in \cite{semsens3}, which maximized the signal-to-clutter-plus-noise ratio (SCNR) under a constant-modulus constraint and optimized target detection performance through energy allocation. Reference \cite{semsens4} addressed the joint optimization of bandwidth allocation for sensing and data transmission in SemCom systems by applying the projected gradient method to maximize semantic spectrum efficiency. Although sensing strategies have advanced noticeably, the optimization goals have, for the most part, been confined to low-level statistical signal metrics such as normalized mean-squared error (NMSE) and peak signal-to-noise ratio (PSNR). Motivated by the need for high-level semantics in intelligent systems, there has been a gradual shift in the focus of radar sensing from low-level sensing metrics to task-oriented design. A framework leveraging graph structures and temporal features was introduced in \cite{semsens6} to perform semantic segmentation on sparse sequential point clouds from millimeter-wave (mmWave) radar, while Reference \cite{semsens7} proposed a semantic simultaneous localization and mapping (SLAM) algorithm using radar-vision fusion to generate semantic grid maps. Likewise, micro-Doppler signatures have been used to perform human activity recognition (HAR) in various scenarios \cite{semsens8}, \cite{semsens9}. 

Reference \cite{semsensmain} introduced the SemSens paradigm based on the information bottleneck principle, in which the authors jointly optimized transmit waveform parameters and receiver representations using a deep learning framework for task-oriented design. SemSens considers a task-oriented representation of the wireless propagation environment by transforming a high-dimensional physical channel observation into a semantic representation \cite{semsens10}. Conventional sensing models explicitly characterize all multipath components through parameters such as delay, Doppler, angle, amplitude, and phase. On the other hand, SemSens retains only those attributes that are relevant to the task. Therefore, the resulting semantic channel can be considered a generalized mapping between the transmitted waveform and the task-relevant representation of the environment, such that irrelevant propagation components are treated as nuisance terms. This semantic extraction forms a many-to-one dimensionality reduction of the physical channel and can eliminate irrelevant channel information while retaining the information required for the sensing task \cite{semsens11}. From an information-theoretic perspective, the data-processing inequality given in \cite{semsensmain} implies that the semantic representation cannot contain more information about the underlying environmental state than the original physical observation. However, the discarded information does not compromise task performance when the extracted information is sufficient for the task. Therefore, SemSens shifts the objective from reconstructing the complete physical channel to extracting a task-relevant representation of the environment.

Despite advances in SemCom and SemSens, existing studies treat them separately in ISAC systems. Semantic processing of communication information and semantic extraction of sensing information have largely been independent processes, even though both tasks share the same propagation environment, waveform, and limited resources. This separation prevents exploiting their inherent interaction at the semantic level. SemCom discards information irrelevant to communication, while SemSens eliminates information unnecessary for sensing. When both coexist, independently determining task-relevant information leads to redundant processing and fails to exploit semantic information shared across tasks. A unified semantic ISAC framework would instead treat communication and sensing as interdependent semantic tasks, jointly determining task-relevant information. This could establish a common semantic representation of the wireless environment and transmitted information, selectively preserving relevant information while discarding the irrelevant. Since the semantic representations required by communication and sensing may overlap significantly, information extracted for one task can benefit the other, improving overall ISAC efficiency. This motivates a unified framework performing both communication and sensing at the semantic level, rather than incorporating semantic processing into only one functionality.

This gap motivates a unified semantic ISAC framework in which SemCom and SemSens are considered jointly. We introduce a semantic ISAC system that simultaneously extracts task-relevant information for both communication and sensing, using a combined encoder for SemCom and SemSens along with adaptive pilot placement and waveform design. We mathematically formulate the joint representation of SemCom and SemSens and propose a pilot selection strategy that facilitates both tasks. Our main contributions are summarized as follows.
\begin{itemize}
    \item We propose SemISAC, a fully semantic approach that unifies SemSens and SemCom and shifts the design objective to task-oriented communication and sensing. Specifically, we formulate the unified SemISAC problem mathematically and establish the joint encoder architecture. In addition, our framework employs joint optimization for the entire pipeline to address the interrelated SemSens and SemCom losses.
    \item We present an adaptive pilot strategy in order to maximize the performance of the SemISAC system. For this purpose, we frame the pilot selection mathematically and adopt a straight-through estimator strategy to facilitate the flow of gradients and train the adaptive pilot selection network.
    \item We instantiate the proposed SemISAC framework and evaluate its performance across multiple communication and sensing tasks. We validate SemCom and SemSens separately with communication and sensing baselines and compare the performance of the proposed approach across different power and subcarrier resources.
\end{itemize}



\begin{figure*}
    \centering
    \includegraphics[width=1\linewidth]{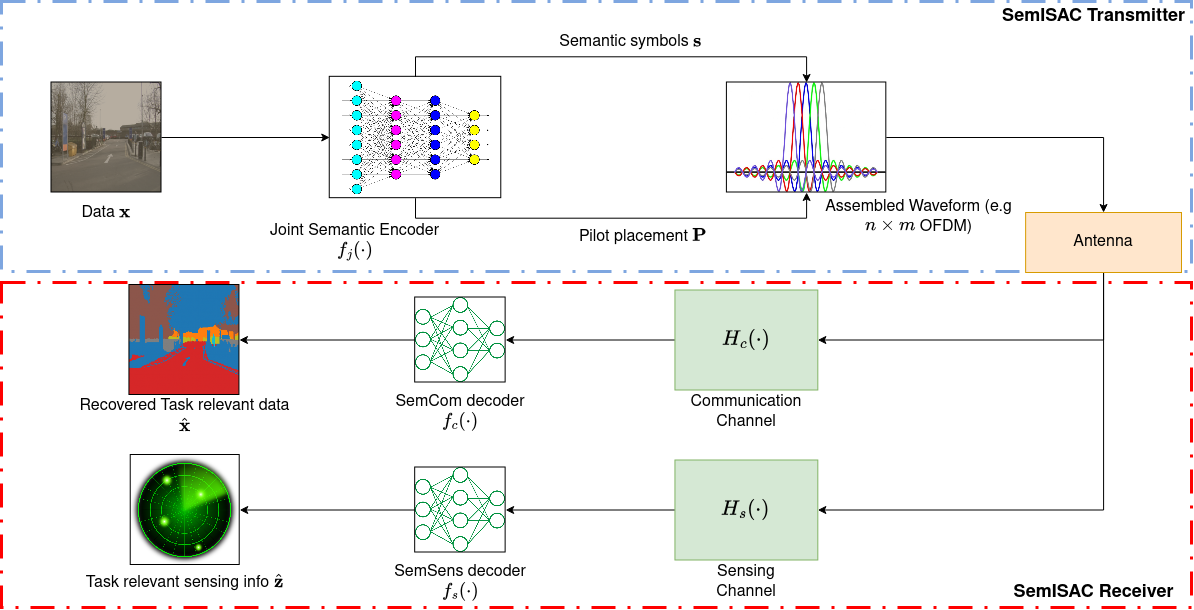}
    \caption{Proposed SemISAC system.}
    \label{fig:placeholder}
\end{figure*}

\section{Problem Formulation and Design Principle}
\label{formulation}
This section formulates the proposed semantic ISAC framework as a joint optimization problem involving SemCom, SemSens, and resource allocation. First, the joint SemCom-SemSens system model is established, where a common semantic representation is designed to support both communication and sensing tasks. The pilot selection problem is then formulated to characterize the tradeoff between sensing performance and the communication resources available in the TF grid, motivating an adaptive allocation strategy. A joint representation learning principle is then developed in which the semantic representation and resource allocation are optimized according to their collective contribution to downstream task performance. Together, these formulations provide the theoretical foundation for the SemISAC system with adaptive pilot selection developed in the subsequent sections.

\subsection{Joint Semantic ISAC Problem}
We consider an ISAC system designed to jointly support both SemCom and SemSens, rather than being dedicated to a single functionality. The proposed framework consists of a joint semantic encoder, a channel model, and separate decoders for SemCom and SemSens. Let $\mathbf{x}$ represent the real-valued source data or observation that needs to be communicated and $\mathbf{z}$ denote the environmental state relevant to sensing. The joint encoder is used to generate the complex-valued semantic representation of the source information
\begin{equation}
    \mathbf{s} = f_j(\mathbf{x}),
    \label{eq:eq1}
\end{equation}
where $f_j(\cdot)$ represents the joint encoder. This representation serves as the common transmitted representation for both communication and sensing and is processed through the relevant channels to produce the observations $Y_c$ and $Y_s$:
\begin{equation}
    Y_c = H_c(\mathbf{s}), \qquad Y_s = H_s(\mathbf{s},\mathbf{z}),
\end{equation}
where $H_c(\cdot)$ and $H_s(\cdot)$ represent the communication and sensing channels, respectively. After being passed through the communication decoder $f_c(\cdot)$ and the sensing decoder $f_s(\cdot)$, respectively, the task outputs are
\begin{equation}
\hat{\mathbf{x}}= f_c(Y_c), \qquad \hat{\mathbf{z}} = f_s(Y_s),
\end{equation}
where $\hat{\mathbf{x}}$ and $\hat{\mathbf{z}}$ are the recovered task-relevant SemCom and SemSens information, respectively. A key distinction of this framework is that neither task requires the reconstruction of the complete transmitted waveform or channel. Instead, each decoder is optimized to extract the semantic information relevant to its task. Consequently, the fundamental objective is to design a unified transmitted representation that simultaneously preserves the semantic information required for communication and the task-oriented physical information required for sensing. This formulation highlights the central challenge of semantic ISAC: learning a common representation that effectively captures and balances the information requirements of both SemCom and SemSens while avoiding unnecessary transmission of task-irrelevant information.

\subsection{Pilot Selection}
In ISAC systems, sensing relies on reference signals, commonly referred to as pilot signals, whose known waveforms and transmission parameters allow the receiver to estimate the wireless channel. In SemSens, the role of pilot signals is even more important because the amount of information extracted from the received signal is deliberately constrained by the semantic objective. Therefore, pilot placement must be designed to provide the SemSens decoder with the information most relevant to the task, while suppressing information arising from irrelevant channel variation and noise. Since the total number of resource elements (REs) in a single waveform is limited, allocating more pilots for sensing necessarily reduces the REs available for communication, creating a resource-allocation problem.

Let $P_i \in \{0, 1\}$ denote the allocation of a resource element indexed by resource unit $i$, where
\begin{equation}
P_{i} =
\begin{cases}
1, & \text{ allocated for sensing},\\
0, & \text{ allocated for communication},
\end{cases}
\end{equation}
Accordingly, the total number of resources allocated to sensing can be expressed as 
\begin{equation}
    N_p(\mathbf{P}) = \sum_{i} P_{i},
\end{equation}
while the remaining resources are available for communication
\begin{equation}
    N_d(\mathbf{P}) = \sum_{i} (1 - P_{i}).
\end{equation}
Thus the pilot placement matrix $\mathbf{P}$ directly determines the division of available time frequency resources between sensing and communication. To capture the fact that the ultimate objective is task performance rather than conventional physical layer accuracy, we define task oriented quality measures for communication and sensing. In particular, let
\begin{equation}
    Q_c(\mathbf{P}) = -\mathcal{L}_c(\mathbf{P}),
\end{equation}
where $\mathcal{L}_c$ denotes the communication related semantic loss such as the error associated with semantic reconstruction or decoding. Similarly,
\begin{equation}
    Q_s(\mathbf{P}) = -\mathcal{L}_s(\mathbf{P}),
\end{equation}
where $\mathcal{L}_s$ represents the sensing loss which may quantify errors in sensing tasks such as object classification, localization or range estimation. Defining the quality measures in terms of downstream task losses allows the value of a pilot resource to be evaluated according to the information it provides for the intended SemSens and SemCom tasks. The optimal pilot allocation can therefore be formulated abstractly as
\begin{equation}
    \mathbf{P}^\star = \arg\max_{\mathbf{P}} \mathcal{U}\bigl(Q_c(\mathbf{P}), Q_s(\mathbf{P})\bigr),
\end{equation}
subject to the available resource and power constraints. A simple realization of the utility function is a weighted combination of the two task oriented quality measures,
\begin{equation}
    \mathcal{U} = \alpha Q_c + (1-\alpha)Q_s,\qquad 0 \leq \alpha \leq 1,
\end{equation}
where $\alpha$ determines the relative importance of SemCom and sensing. More generally, $\mathcal{U}(\cdot)$ can be selected to represent application specific priorities or quality of service requirements.

The above formulation highlights an important property of the pilot placement problem that the optimal allocation is not necessarily fixed across different channel realizations or task requirements. The usefulness of a particular resource element for sensing depends on the underlying channel and propagation conditions, while its usefulness for SemCom depends on the information required by the downstream semantic task. Consequently, the optimal allocation can be expressed as a function of the current system state,
\begin{equation}
    \mathbf{P}^\star = \mathbf{P}^\star(\mathbf{x}, \mathbf{z})
\end{equation}
Therefore, a fixed pilot pattern designed independently of the instantaneous channel and task conditions cannot, in general, achieve the optimal trade off between sensing and SemCom. Instead, the pilot configuration should adapt to the current system state so that sensing resources are concentrated where they provide the greatest task oriented benefit, while unnecessary pilot resources are avoided.

This state dependence can be made explicit by considering the effect of the
instantaneous SNR on the sensing-communication resource tradeoff.
\begin{assumption}[Noise limited pilot count model]\label{as:linear-pilots}
In the noise-limited regime, the sensing distortion decreases proportionally
with the effective pilot SNR resource $N_p\gamma$, while the communication
distortion decreases proportionally with the rate $\log(1+\gamma)$ available
to the remaining $N_c-N_p$ symbols. Specifically, there exist constants
$a,c>0$, independent of $N_p$ and $\gamma$, such that
\begin{equation}
D_s(N_p,\gamma)=\frac{a}{N_p\gamma},\qquad
D_c(N_p,\gamma)=\frac{c}{(N_c-N_p)\log(1+\gamma)}.
\end{equation}
\end{assumption}

Under Assumption~\ref{as:linear-pilots}, the sensing distortion decreases
proportionally to $1/\gamma$ while the communication distortion decreases
more slowly, proportionally to $1/\log(1+\gamma)$. Consequently, the marginal
benefit of allocating an additional resource element to sensing changes with
SNR, and so does the optimal sensing resource count.
\begin{proposition}[Optimal pilot count decreases with SNR]\label{thm:npilots}
Under Assumption~\ref{as:linear-pilots}, for any $w\in(0,1)$ the weighted
distortion $J(N_p)=w D_s(N_p,\gamma)+(1-w)D_c(N_p,\gamma)$ has a unique
minimizer
\begin{equation}
N_p^\star(\gamma)=N_c\,
\frac{\sqrt{wa/\gamma}}{\sqrt{wa/\gamma}+\sqrt{(1-w)c/\log(1+\gamma)}},
\label{eq:npstar}
\end{equation}
which is strictly decreasing in $\gamma$, with $N_p^\star(\gamma)\to0$ as
$\gamma\to\infty$. The proof is given in Appendix~\ref{app:proof-npilots}.
\end{proposition}

This result provides a theoretical justification for making pilot allocation
SNR-aware. As the channel becomes less noise-limited, fewer pilot resources
are required to meet the sensing objective, leaving more resource elements
available for SemCom.

However, directly solving the above binary optimization problem for every channel realization can be computationally expensive, particularly when the number of available TF resources is large. The monotonic SNR dependence established by Proposition~\ref{thm:npilots} further suggests that the pilot allocation policy should explicitly account for the operating SNR rather than relying on a fixed pilot count. This motivates a learning based approach in which a neural network learns an adaptive pilot allocation policy. Specifically, the allocation can be represented as
\begin{equation}
    \mathbf{P} = f_{\boldsymbol{\theta}}(\mathbf{x},\mathbf{z}),
    \label{eq:eq2}
\end{equation}
where $f_{\boldsymbol{\theta}}(\cdot)$ denotes the proposed neural network based allocation function parameterized by $\boldsymbol{\theta}$. The network therefore learns to map the observed system state to an appropriate pilot configuration,
\begin{equation}
    f_{\boldsymbol{\theta}}(\mathbf{x},\mathbf{z}) \approx \mathbf{P}^\star(\mathbf{x},\mathbf{z}).
\end{equation}

In this way, the neural network does not replace the underlying resource-allocation objective rather it provides a computationally efficient approximation to the optimal adaptive policy. The resulting pilot placement mechanism can dynamically adjust the sensing resources according to channel conditions and the relative requirements of the sensing and SemCom tasks, thereby providing the theoretical basis for the proposed adaptive pilot selection framework.

\subsection{Joint Representation Principle}
The preceding formulations establish that the transmitted signal must preserve the semantic information required for communication while simultaneously preserving sufficient task-relevant information for sensing. These objectives are coupled through the finite TF resources available for transmission. Consequently, the central design problem is not simply to construct an encoder for each task independently, but to learn a unified representation that jointly serves the heterogeneous requirements of SemCom and SemSens.

Let $\mathbf{s}$ denote the latent semantic representation extracted from the source data, as in \eqref{eq:eq1}, and let $\mathbf{P}$ be a matrix that denotes the sensing-resource allocation pattern over the available TF resource elements as before. The resulting transmitted TF representation $\mathbf{X}$ can be abstractly written as
\begin{equation}
\mathbf{X} = \mathcal{A}(\mathbf{s},\mathbf{P}),
\end{equation}
where $\mathcal{A}(\cdot)$ represents the mechanism that maps the semantic representation and the selected sensing resources onto the transmitted TF grid. Here, $\mathbf{s}$ captures the information that should be conveyed through the waveform, whereas $\mathbf{P}$ determines where explicit sensing resources are introduced. The joint representation principle can therefore be formulated as a multi-objective optimization problem. Specifically, the parameters of the semantic encoder and task-specific decoders, together with the sensing-resource allocation, are jointly optimized according to
\begin{equation}
\min_{\boldsymbol{\theta},\boldsymbol{\phi},\boldsymbol{\psi},\mathbf{P}}
\quad
\lambda_{\mathrm{c}}\mathcal{L}_{\mathrm{c}}(\mathbf{P})
+
\lambda_{\mathrm{s}}\mathcal{L}_{\mathrm{s}}(\mathbf{P})
+
\lambda_{\mathrm{r}}\mathcal{R}(\mathbf{P}),
\label{eq:joint_representation_objective}
\end{equation}
where $\boldsymbol{\theta}$, $\boldsymbol{\phi}$, and $\boldsymbol{\psi}$ denote the parameters associated with the joint semantic encoder, SemCom decoder, and SemSens decoder, respectively. The terms $\mathcal{L}_{\mathrm{c}}$ and $\mathcal{L}_{\mathrm{s}}$ denote the SemCom and SemSens losses, while $\mathcal{R}(\mathbf{P})$ represents the cost associated with allocating sensing resources. The coefficients $\lambda_{\mathrm{c}}$, $\lambda_{\mathrm{s}}$, and $\lambda_{\mathrm{r}}$ control the relative importance of communication performance, sensing performance, and resource efficiency.

This formulation establishes the fundamental tradeoff underlying the proposed semantic ISAC framework. A representation optimized solely for $\mathcal{L}_{\mathrm{c}}$ may discard physical information that is useful for sensing, whereas a representation optimized solely for $\mathcal{L}_{\mathrm{s}}$ may allocate excessive resources to sensing or preserve information that is unnecessary for communication. Joint optimization instead encourages the learned representation to retain information that is simultaneously valuable to both tasks while avoiding task-irrelevant information and unnecessary sensing overhead.

The proposed architecture can then be interpreted as one realization of the optimization problem in \eqref{eq:joint_representation_objective}. The joint semantic encoder generates the common representation $\mathbf{s}$ from the source information, as in \eqref{eq:eq1}, while the adaptive pilot selection mechanism determines the sensing resource configuration as in \eqref{eq:eq2}. These components jointly determine the transmitted TF representation
\begin{equation}
\mathbf{X}=\mathcal{A}\left(
f_{j}(\mathbf{x}),
f_{\boldsymbol{\theta}}(\mathbf{x},\mathbf{z})
\right) = \mathcal{A}\left(
\mathbf{s},
\mathbf{P}
\right).
\end{equation}
After propagation through the communication and sensing channels, task-specific decoders extract the required semantic information from the received observations.

Thus, the proposed framework constitutes a joint representation learning problem in which the semantic representation and sensing-resource allocation are optimized according to their collective contribution to downstream task performance. The TF waveform consequently becomes a shared information-bearing structure. The data REs convey semantic information for communication, while its sensing REs provide additional physical information required for environmental inference. Consequently, the optimal representation is inherently task- and resource-dependent. In general, there is no universally optimal semantic representation or fixed pilot pattern that is simultaneously optimal for all communication and sensing conditions. The desired representation depends on the source information, environmental state, channel conditions, and relative importance assigned to the two tasks. The joint objective in \eqref{eq:joint_representation_objective} therefore provides a general mathematical principle for learning representations that adapt to these requirements. Accordingly, the key principle of the proposed semantic ISAC framework is that the same transmitted TF representation should be jointly optimized to preserve communication-relevant semantics and sensing-relevant physical information, subject to the finite sensing and communication resources available in the waveform. The neural architecture and adaptive pilot selection mechanism introduced in the subsequent sections provide a practical realization of this principle.

\begin{figure}[h]
    \centering
    \includegraphics[width=1\linewidth]{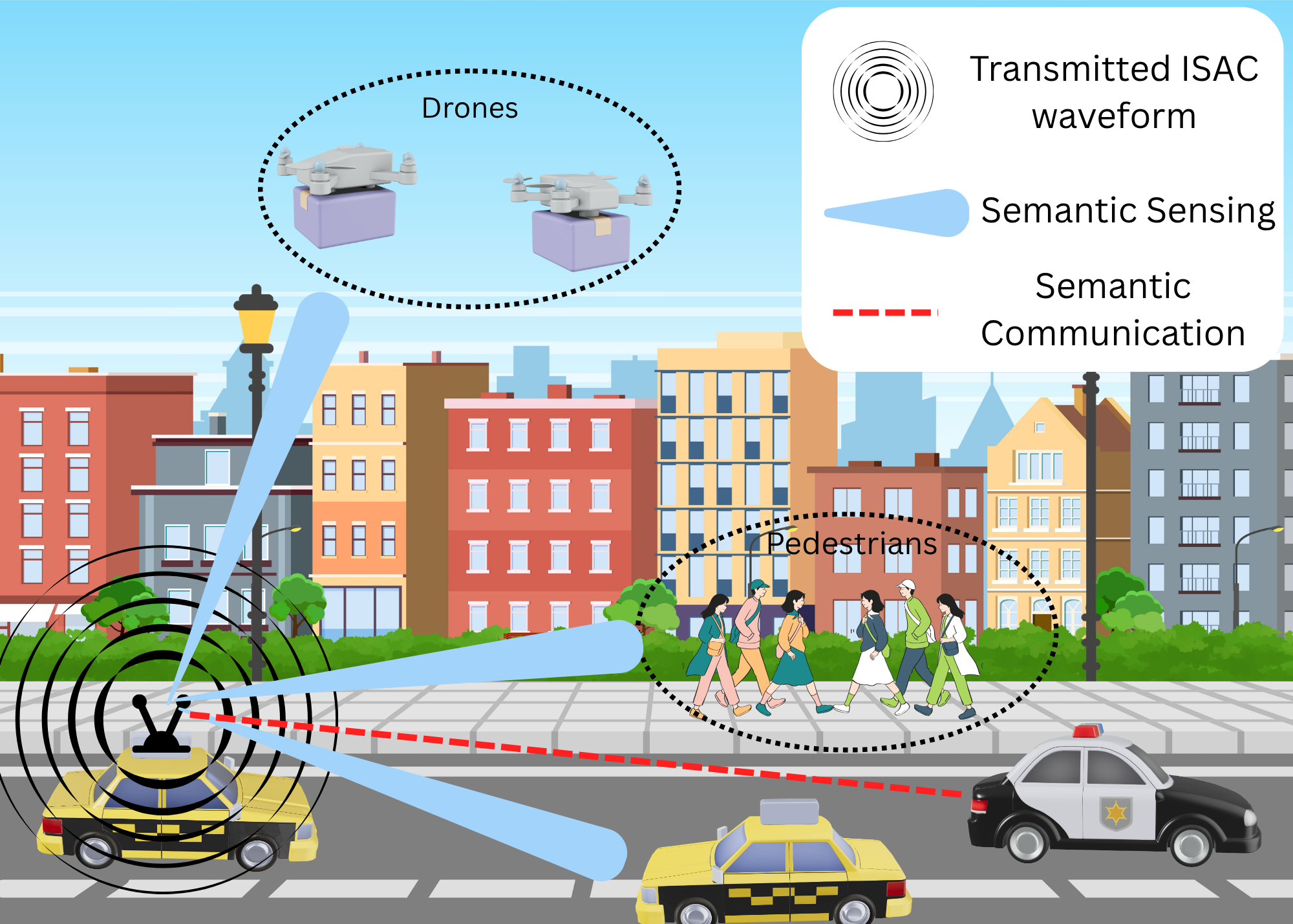}
    \caption{System model.}
    \label{fig:scenario}
\end{figure}

\section{System Model}
\label{sysmod}
In this study, we consider the vehicular ISAC scenario as illustrated in Figure \ref{fig:scenario}. The scenario uses an autonomous car equipped with an ISAC transceiver, driving through an urban road environment. The car perceives the surrounding
road scene and, to support cooperative and safe driving, it must both (i) share this
perception with the nearby \emph{neighboring vehicles} and (ii) remain aware of the
potential obstacles in its vicinity, such as pedestrians, cars, and drones. To perform these functions, the vehicle emits a \emph{single} dual-function OFDM waveform. The waveform is structured as a time-frequency (TF) resource grid, and its cells are divided into two categories: \emph{pilot} cells and \emph{data} cells. This OFDM waveform simultaneously performs two tasks:
\begin{enumerate}
\item \textbf{SemCom:} it transmits a semantic description of the road scene to the neighboring vehicles, which recover pixel-wise segmentation of the image rather than the original image. 
\item \textbf{SemSens:} the backscattered signal of the transmitted waveform is processed at the transmitting vehicle to classify the object (pedestrian,
car, or drone) and to estimate its distance.
\end{enumerate}

\subsection{Dual-Function OFDM Frame and Power Allocation} \label{sec:frame}
\subsubsection{Time--frequency grid}
The dual-function waveform is an OFDM frame whose resources are naturally organized on a two-dimensional \emph{time--frequency} (TF) grid. The frame
consists of $M$ orthogonal subcarriers along the frequency axis and $N$ OFDM symbols along the \emph{slow-time} axis, giving a total of 
\begin{equation}
N_c = N \times M
\end{equation}
resource cells. Each
cell, indexed by $(n,m)$ with $n\in\{0,\dots,N-1\}$ and
$m\in\{0,\dots,M-1\}$, carries a single complex symbol on the $m$-th
subcarrier of the $n$-th OFDM symbol. The subcarriers are separated by
$\Delta f = B/M$, where $B$ is the total occupied bandwidth and the OFDM symbols are transmitted at a slow-time  interval $T$. The simulation operates at carrier frequency $f_c$ with wavelength $\lambda = c/f_c$, where
$c$ is the speed of light. The numerical values of these parameters used in the study are
listed in Table~\ref{tab:params}.

The two axes of the time-frequency grid encode complementary parameters that can be utilized by the sensing receiver to classify the target. Consider a single reflecting path with round-trip delay $\tau$ and Doppler shift $\nu$. The delay $\tau$ produces a phase that linearly changes along the subcarriers (frequency axis), $e^{-j2\pi m\,\Delta f\,\tau}$. Similarly, the Doppler shift $\nu$ produces a phase that varies linearly along the slow-time axis $e^{\,j2\pi \nu\, nT}$. The target distance $R=c\tau/2$ is estimated along the frequency axis, whereas its velocity $v=\lambda\nu/2$ and the micro-Doppler modulation of its moving parts (if any) are estimated along the slow-time axis $n$. This delay--Doppler separation allows the pilot cells to be utilized for sensing while the data cells convey the semantic data. 
\subsubsection{Pilot--data partitioning}
We introduce deep-learning-based pilot selection using a binary \emph{pilot-placement mask} $\mathbf{P}\in\{0,1\}^{N\times M}$. The mask is not fixed; instead, a deep learning encoder assigns a \emph{selection score} to each cell and then compares the score against a threshold. Cells whose scores exceed the threshold are selected as pilot cells $[\mathbf{P}]_{n,m}=1$, and the rest function as data cells $[\mathbf{P}]_{n,m}=0$. Importantly, the threshold is a function of SNR, so the resulting number of pilots $N_p$ adapts to the channel conditions.

\subsubsection{Power allocation and transmit grid}
The transmitting vehicle operates under a fixed total transmit power $P_t$ per OFDM frame, which is shared between the pilot cells and data cells through a single \emph{power-allocation coefficient}
$\rho\in(0,1)$. This $\rho$ is the fraction of the total power assigned to the pilots. The resulting per-cell transmit energy is given by
\begin{equation}
E_p=\frac{\rho\,P_t}{N_p},
\qquad
E_d=\frac{(1-\rho)\,P_t}{N_d},
\label{eq:budget}
\end{equation}
 such that the pilot cells altogether receive a power budget of $E_pN_p=\rho P_t$ and the
data cells $E_dN_d=(1-\rho)P_t$.
Consequently, the total power budget constraint
\begin{equation}
E_pN_p+E_dN_d=P_t
\label{eq:budget-constraint}
\end{equation}
holds for any value of $\rho$ and any number of pilot cells $N_p$. The coefficient $\rho$ therefore controls the communication-sensing power tradeoff in the system, where a larger value of $\rho$ allocates more power to the pilot cells, improving sensing and channel estimation at the cost of weaker semantic data symbols, and vice versa. Under the noise-limited regime, this tradeoff admits a closed-form
characterization.

\begin{assumption}[Noise-limited power split model]\label{as:linear-power}
In the noise-limited regime, the sensing and communication task distortions
are inversely proportional to their respective allocated energies.
Specifically, there exist constants $a,b>0$, independent of $\rho$, such that
\begin{equation}
D_s(\rho)=\frac{a}{\rho},\qquad D_c(\rho)=\frac{b}{1-\rho},\qquad 0<\rho<1.
\end{equation}
\end{assumption}

\begin{remark}
Assumption~\ref{as:linear-power} follows from \eqref{eq:budget}. The pilot
estimation error variance scales as $\sigma^2/(E_pN_p)=\sigma^2/(\rho P_t)$,
giving $a=\kappa_s\sigma^2/P_t$. Similarly, the per-symbol data distortion
scales as $\sigma^2/E_d=N_d\sigma^2/((1-\rho)P_t)$, giving
$b=\kappa_c N_d\sigma^2/P_t$, with task-dependent $\kappa_s,\kappa_c>0$.
\end{remark}

\begin{lemma}[Unique Pareto-optimal power split]\label{thm:rho}
Under Assumption~\ref{as:linear-power}, let $w\in(0,1)$ denote the priority
assigned to the sensing task. The weighted distortion
$J(\rho)=wD_s(\rho)+(1-w)D_c(\rho)$ is strictly convex on $(0,1)$ and attains
a unique minimizer
\begin{equation}
\rho^\star(w)=\frac{\sqrt{wa}}{\sqrt{wa}+\sqrt{(1-w)b}}.
\label{eq:rhostar}
\end{equation}
Moreover, $\rho^\star$ is strictly increasing in $w$, with $\rho^\star\to0$
as $w\to0$ and $\rho^\star\to1$ as $w\to1$. The proof is given in
Appendix~\ref{app:proof-rho}.
\end{lemma}

A larger value of $\rho$ therefore allocates more power to the pilot cells,
improving sensing and channel estimation at the cost of weaker semantic data
symbols, and vice versa.

Combining the pilot placement $\mathbf{P}$ with the power allocation budget, the transmitted grid
$\mathbf{X}\in\mathbb{C}^{N\times M}$ is structured as follows:
\begin{equation}
\mathbf{X}
=\underbrace{\sqrt{E_p}\,s_p\,\mathbf{P}}_{\text{pilots (sensing)}}
+\underbrace{\sqrt{E_d}\,(\mathbf{1}-\mathbf{P})\odot\mathbf{S}}_{\text{data (communication)}},
\label{eq:grid}
\end{equation}
where $s_p$ represents the known pilot symbol with unit magnitude ($|s_p|=1$), and $\mathbf{1}$ is the all-ones matrix. $\mathbf{S}\in\mathbb{C}^{N\times M}$ is
the grid of SemCom symbols produced by the deep learning encoder. 

\subsection{Semantic Communication Model} \label{sec:comm}
To convey the road scene efficiently, we use SemCom instead of sharing raw images. The transmitting vehicle encodes only the semantic description required by
the neighboring vehicle to reproduce a pixel-wise segmentation of the scene.
\subsubsection{Semantic symbol generation}
The transmitting vehicle perceives the roadside environment through an image $\mathbf{x}\in\mathbb{R}^{3\times H\times W}$. A Swin-Transformer is used as a semantic encoder to convert the image into semantic symbols, which are complex numbers carrying the task-relevant semantics of the scene.
\begin{equation}
\mathbf{S}=\mathcal{E}_{\boldsymbol\theta}(\mathbf{x})\in\mathbb{C}^{N\times M}.
\label{eq:symbols}
\end{equation}
Because the encoder's output has a nonuniform magnitude, it is rescaled to unit average power, $\frac{1}{N_c}\sum_{n,m}|[\mathbf{S}]_{n,m}|^2=1$.
After this normalization, the average power of each data cell equals $E_d$, such that
data-cell transmit power is determined solely by the coefficient $E_d$
rather than by the image-dependent symbol magnitudes.
\subsubsection{Communication channel}
These semantic symbols are then transmitted to the neighboring vehicle, which receives the data from a \emph{direct} (line-of-sight, LOS) path and several
\emph{reflected} (non-line-of-sight, NLOS) paths that bounce off the surrounding objects. Their combination is modeled as a frequency-selective Rician channel.

\begin{equation}
\mathbf{H}^{\mathrm c}=\sqrt{\frac{K}{K+1}}\,\mathbf{H}^{\mathrm{LOS}}+\sqrt{\frac{1}{K+1}}\,\mathbf{H}^{\mathrm{NLOS}}
\label{eq:hcomm}
\end{equation}
where the Rician factor $K$ is used to control power in the direct path relative
to the reflected ones. The \emph{direct} path travels straight from the transmitting vehicle to the
neighboring vehicle without reflecting off anything
\begin{equation}
[\mathbf{H}^{\mathrm{LOS}}]_{n,m}=e^{\,j\phi_0}\,e^{\,j2\pi\nu_0 nT},
\label{eq:hlos}
\end{equation}
where $e^{\,j\phi_0}$ is a constant random phase and $e^{\,j2\pi\nu_0 nT}$ is a slow phase rotation along the slow-time axis caused by the bulk-Doppler shift  $\nu_0$ from the relative motion of two vehicles. 
The reflected signals are collected and summed before reaching the neighboring vehicle
\begin{equation}
\begin{aligned}
[\mathbf{H}^{\mathrm{NLOS}}]_{n,m}
&=\sum_{l=1}^{L_c} a_l\,e^{-j2\pi m\,\Delta f\,\tau_l}\,e^{\,j2\pi\nu_l nT},\\
&\quad\ \, a_l\sim\mathcal{CN}(0,\,1/L_c).
\end{aligned}
\label{eq:hnlos}
\end{equation}
Each path $l$ has a complex gain $a_l$ (a random amplitude and phase), a delay $\tau_l$ that produces a linear phase change across subcarriers, and a Doppler $\nu_l$ that produces a phase variation across the slow-time axis. Overall, a large $K$ means that the direct path dominates and the channel is smooth and clean, whereas a small $K$ means that the reflected signals dominate. 
\subsubsection{Semantic decoding}
The signal received by the neighboring vehicle is
\begin{equation}
\mathbf{Y}^{\mathrm c}=\mathbf{H}^{\mathrm c}\odot\mathbf{X}+\mathbf{N}^{\mathrm c},
\label{eq:yc}
\end{equation}
where the entries of $\mathbf{N}^{\mathrm c}$ are independent and identically distributed additive white Gaussian noise samples. Since the pilot and data cells experience the same channel within the grid, channel state information (CSI) can be estimated from the pilot cells. The received signal $\mathbf{Y}^{\mathrm c}\odot(\mathbf{1}-\mathbf{P})$, along with the channel estimate $\widehat{\mathbf{H}}^{\mathrm c}$, is passed to a semantic decoder $\mathcal{D}_{\boldsymbol\phi}(\cdot)$
which recovers the segmentation of the
road scene.
\begin{equation}
\hat{\mathbf{m}}
=\mathcal{D}_{\boldsymbol\phi}\big(\mathbf{Y}^{\mathrm c},\widehat{\mathbf{H}}^{\mathrm c},\mathbf{P}\big)
\in\{1,\dots,O_s\}^{H\times W}.
\label{eq:seg}
\end{equation}
Here, $O_s$ is the number of semantic classes on which the model is trained, and each pixel of $\hat{\mathbf{m}}$
is assigned one class label. This segmentation can be used by the vehicle to interpret the surrounding obstacles from another perspective.
\subsection{Semantic Sensing Model} \label{sec:sensing}
For the sensing part of semantic ISAC, the transmitting vehicle processes the echo of its own waveform to identify surrounding objects and measure their distances.

\subsubsection{Sensing channel}
The transmitted OFDM grid reflects off moving targets and returns to the transmitting vehicle. The target obstacle lies at a range of $R$ and moves with radial velocity $v$. Instead of acting as a single point, the target reflects the signal from $L$ different points on its body. The target as a whole moves at a certain velocity; however, its own parts move independently, like a person's swinging arms and legs or a
drone's spinning blades. This extra motion of the parts adds a small
time-varying pattern to the reflected signal called \emph{micro-Doppler}. The sensing channel is
\begin{equation}
[\mathbf{H}^{\mathrm s}]_{n,m}
=\sum_{l=1}^{L}\alpha_l\,
e^{-j2\pi m\,\Delta f\,\tau}\,
e^{\,j2\pi \nu\, nT}\,
e^{\,j\beta_l \sin(2\pi f_{\mu,l}\, nT+\varphi_l)},
\label{eq:hsens}
\end{equation}
Each of the exponential terms plays a different role. The delay term
$e^{-j2\pi m\,\Delta f\,\tau}$ which changes linearly across the subcarriers, encodes
the \emph{range} of the target, $e^{\,j2\pi \nu\, nT}$ changes
linearly across the slow-time axis and shows the \emph{velocity} of the target, and $e^{\,j\beta_l \sin(2\pi f_{\mu,l}\, nT+\varphi_l)}$ gives the micro-Doppler along the slow-time axis. $\tau = 2R/c$ is the round-trip delay and $\nu = 2v/\lambda$ is the bulk Doppler caused by the velocity $v$ of the target. These two are the same for all $L$ points. The term $\alpha_l\sim\mathcal{CN}(0,1/L)$ is a complex gain that indicates the reflection strength of point $l$, while $\beta_l$, $f_{\mu,l}$, and $\varphi_l$ give the strength, frequency, and phase of its micro-Doppler, respectively.

\begin{table}[t]
\centering
\caption{Simulation and model parameters.}
\label{tab:params}
\scriptsize
\setlength{\tabcolsep}{3.5pt}
\renewcommand{\arraystretch}{0.85}
\begin{tabular}{lll}
\hline
\textbf{Symbol} & \textbf{Parameter} & \textbf{Value}\\
\hline
\multicolumn{3}{l}{\textit{Grid \& Channel}}\\
$N$ & OFDM symbols & 32\\
$M$ & Subcarriers & 64\\
$N_c$ & Total cells & 2048\\
$B$ & Bandwidth & 30 MHz\\
$\Delta f$ & Subcarrier spacing & 468.75 kHz\\
$f_c$ & Carrier frequency & 3 GHz\\
$\lambda$ & Wavelength & 0.1 m\\
$T$ & Repetition interval & 1 ms\\
$L$ & Paths/target & 3\\
$K$ & Rician factor & 6 dB\\
$L_c$ & Comm. reflected paths & 3\\
$O_s$ & Segmentation classes & 11\\
$H\times W$ & Image resolution & $256\times256$\\
\hline
\multicolumn{3}{l}{\textit{Optimization}}\\
$\xi$ & Learning rate & $10^{-3}$\\
$N_b$ & Batch size & 16\\
$N_e$ & Epochs & 40\\
$\kappa$ & Gradient clip & 1.0\\
$\gamma_{\rm tr}$ & Training SNR & $[-6,20]$ dB\\
$\gamma_{\rm ev}$ & Evaluation SNR & $[-9,20]$ dB\\
\hline
\multicolumn{3}{l}{\textit{Architecture}}\\
$d_f$ & Feature dimension & 8192\\
$Q$ & Range bins & 256\\
$N_r$ & Ranging projections & 4\\
$\omega$ & Gate temperature & 0.1\\
\hline
\multicolumn{3}{l}{\textit{Objective}}\\
$\gamma_{\rm cls}$ & Recognition penalty & 0.10\\
$\gamma_{\rm rng}$ & Ranging penalty & 0.05\\
$\sigma_R$ & Range normalization & 57.7 m\\
$s_i$ & Task loss weights & Learned\\
\hline
\multicolumn{3}{l}{\textit{Sensing / Signal}}\\
$R_{\max}$ & Max. range & 200 m\\
$\beta$ & Micro-Doppler depth & 2.0 rad\\
$N_t$ & Targets/frame & 1\\
$N_p$ & Pilot budget & 256\\
$s_p$ & Pilot symbol & 1\\
$E_p,E_d$ & Pilot/data energy & 1\\
$\rho$ & Power split (evaluation) & $[0.05,0.95]$\\
\hline
\end{tabular}
\end{table}

\subsubsection{Echo processing and target recognition}
The reflected signal is received by the transmitting vehicle 
\begin{equation}
\mathbf{Y}^{\mathrm s}=\mathbf{H}^{\mathrm s}\odot\mathbf{X}+\mathbf{N}^{\mathrm s}.
\label{eq:ys}
\end{equation}
Only pilot signals are utilized for SemSens and data cells are discarded using the mask $\mathbf{P}$.
The receiver measures how the target changed the pilots by dividing the received pilot signals by the original ones. 
\begin{equation}
\widehat{\mathbf{H}}^{\mathrm s}
=\big(\mathbf{Y}^{\mathrm s}\oslash(\sqrt{E_p}\,s_p)\big)\odot\mathbf{P}.
\label{eq:hs-est}
\end{equation}
Conventional radars reconstruct the target's description; however, we perform \emph{task-oriented
SemSens}, where two task-oriented deep learning decoders map
the pilot observation \emph{directly} to the information needed by the task
\begin{equation}
    \hat c = \mathcal{G}_{\text{cls}}(\widehat{\mathbf{H}}^{\mathrm{s}}),
    \qquad
    \hat R = \mathcal{G}_{\text{rng}}(\widehat{\mathbf{H}}^{\mathrm{s}}),
\end{equation}
where $\hat{c}\in\{\text{pedestrian},\text{car},\text{drone}\}$ is the category of the target object and $\hat{R}$ is the distance from the receiver to that object.
The semantic decoders only capture the \emph{meaning} of the reflection, such as ``a pedestrian
$30$~m ahead,'' rather than reconstructing a full physical description of it. This makes the
sensing more efficient as no effort is spent on recovering unnecessary details that the task does not need.

\section{Proposed Method}
\label{proposed}
In our proposed approach, the following components are jointly learned: the semantic encoder $\mathcal{E}_{\boldsymbol\theta}$ that produces semantic data symbols, the
pilot-selection mask that decides which cells will act as pilot cells, the
communication decoder $\mathcal{D}_{\boldsymbol\phi}$ that reconstructs the segmentation of the road environment and
the two sensing decoders $\mathcal{G}_{\text{cls}}$ and
$\mathcal{G}_{\text{rng}}$ that classify the target and estimate its range.
\subsection{Training and Evaluation Dataset}
These road scenes are taken from the
Cambridge-driving Labeled Video Database (CamVid)~\cite{camvid}. It consists of images captured from the viewpoint of a driving vehicle, as illustrated in Fig.~\ref{fig:road}. Each image has annotated labels, which provide a pixel-wise segmentation of different objects present in the image,
as illustrated in Fig.~\ref{fig:gt}. Following common practice, we use the $O_s = 11$ semantic classes (such as road, building, car, pedestrian, sign, and pavement), with the remaining pixels marked as unlabeled. These images act as the source $\mathbf{x}$ for the semantic encoder, which converts them into semantic symbols and whose segmentation
is reconstructed by the neighboring vehicles.

\subsection{Semantic Encoder and Task Decoders}
A Swin-Transformer (ST)~\cite{swin} is used as a semantic encoder that extracts semantic features from the road image. For efficient data processing, we use the \emph{Swin-Tiny} variant with a patch size of $4$ and a window size of $7$, initialized with weights pretrained on ImageNet. Only the head of the ST learns and changes its weights during the training procedure, and the backbone weights are frozen. The ST processes the image hierarchically: it first
divides the $256\times256$ image into non-overlapping $4\times4$ patches and
then applies several stages of self-attention and patch-merging layers to gradually reduce the spatial size while increasing the number of feature channels. This produces an $8\times8$ feature map with $768$ channels. Unlike global attention, the ST computes self-attention only within local windows, which maintains efficiency. Inside a window, it has query, key, and value matrices 
$\mathbf{\mathfrak{Q}},\mathbf{\mathfrak{K}},
\mathbf{\mathfrak{V}}$, and the attention is computed as
\begin{equation}
    \mathrm{Attn}(\mathbf{\mathfrak{Q}},\mathbf{\mathfrak{K}},\mathbf{\mathfrak{V}})
    = \mathrm{softmax}\!\left( \frac{\mathbf{\mathfrak{Q}}\mathbf{\mathfrak{K}}^\top}{\sqrt{d}}
      + \mathbf{\mathfrak{B}} \right)\mathbf{\mathfrak{V}},
\end{equation}
where $d$ is the dimension of the features and $\mathbf{\mathfrak{B}}$ is the learned bias. 
At the receiving vehicle, the semantic decoder $\mathcal{D}_{\boldsymbol\phi}$ recovers the segmentation of the road. Unlike the semantic encoder, the decoder uses a convolutional encoder--decoder neural network trained from scratch for this task. It is given the received data cells and the channel estimate calculated using the pilot cells. The first convolutional layers compress the received input; then, a sequence of upsampling layers expands it into a segmentation map. In this way, the decoder learns to decode the road segmentation. At the transmitting vehicle, SemSens is carried out using pilot cells only. The
decoder $\mathcal{G}_{\text{cls}}$ receives the pilot observation
$\widehat{\mathbf{H}}^{\mathrm{s}}$ with its real and imaginary parts and processes it with a
convolutional neural network to classify the target. The information it relies upon is the micro-Doppler pattern created by the movement of different parts of the target, such as a pedestrian's limbs or a drone's rotating blades. 
Ranging works in a different way, as each distance imprints a different pattern on the received pilots. We construct a dictionary of $Q$ candidate distances $R_q$ and their corresponding patterns beforehand and insert them into the ranging decoder:
\begin{equation}
    D_{m,q} = e^{-j 2\pi\, m\, \Delta f\, \tau_q},
    \qquad \tau_q = \frac{2 R_q}{c},
\end{equation}
where $\tau_q$ is the delay of the echo coming back from that distance. The
decoder compares the received pilots against every pattern in the map and gives each one a match score $z_q$. The closer the match, the higher the score and the more likely it is that the target is at that distance.
A small neural network then turns these scores into a single continuous estimate 
\begin{equation}
    \hat{R} = \sum_{q=1}^{Q} \frac{\exp(z_q)}{\sum_{q'}\exp(z_{q'})}\, R_q .
\end{equation}
Only this small network and a short input layer are trained; the distance mappings remain fixed. This method gives accurate distance estimates while needing less data for training. 

\begin{figure*}[t]
    \centering
    
    \begin{subfigure}{0.66\columnwidth}
        \centering
        \includegraphics[width=1\linewidth]{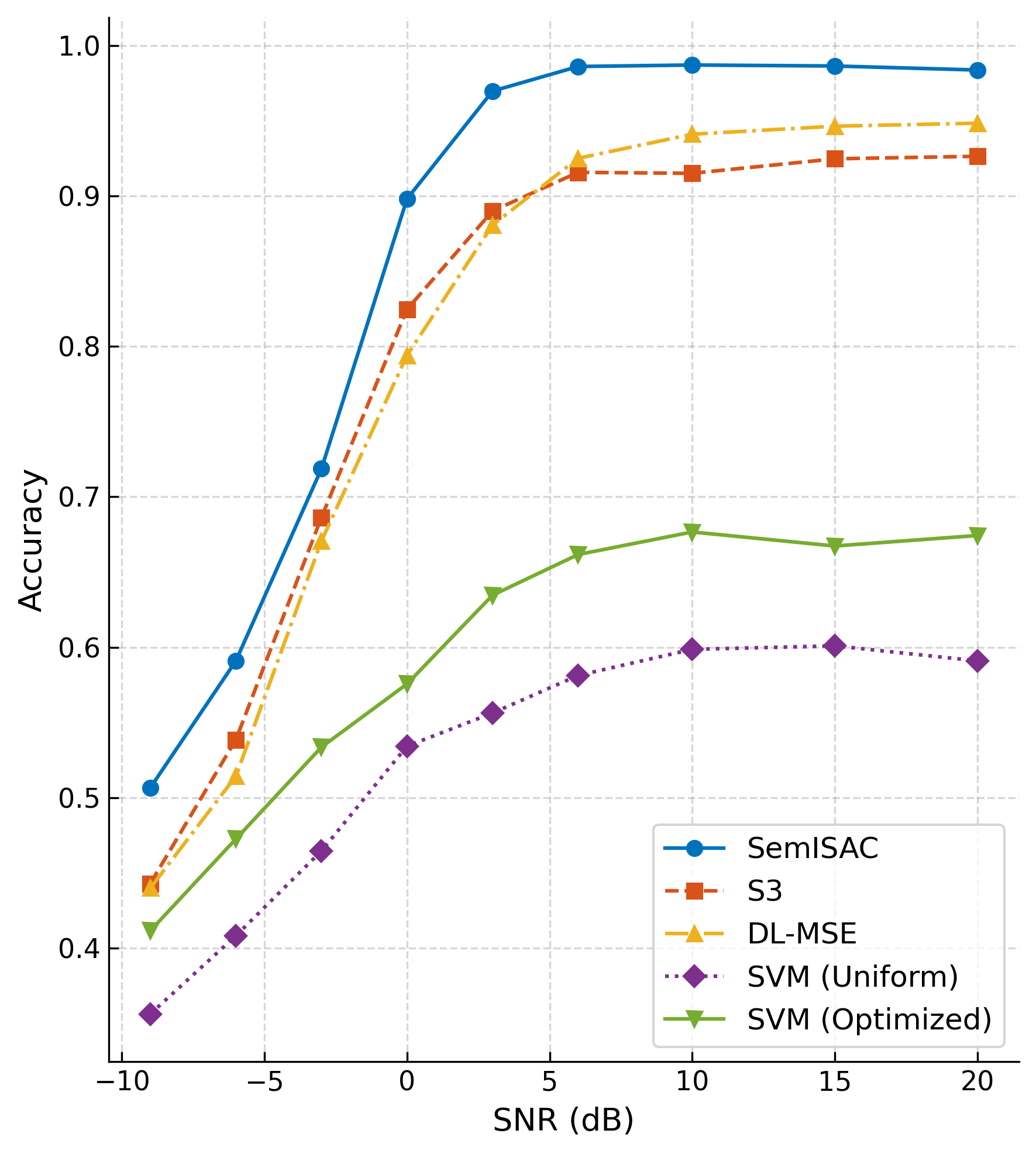}
        \caption{Sensing classification accuracy}
        \label{fig:graph_a}
    \end{subfigure}
    \begin{subfigure}{0.66\columnwidth}
        \centering
        \includegraphics[width=1\linewidth]{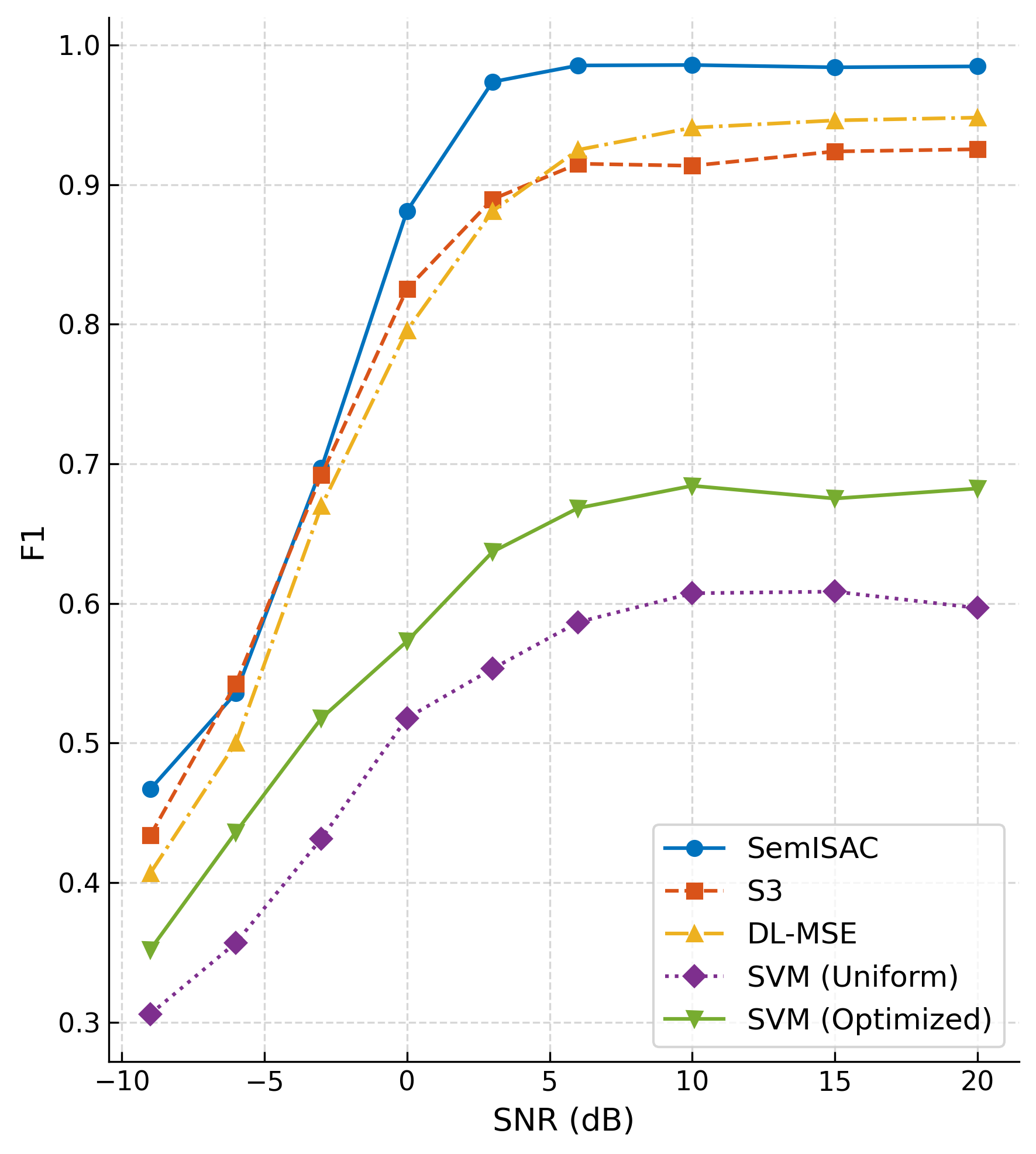}
        \caption{Sensing classification F1}
        \label{fig:graph_b}
    \end{subfigure}
    \begin{subfigure}{0.66\columnwidth}
        \centering
        \includegraphics[width=1\linewidth]{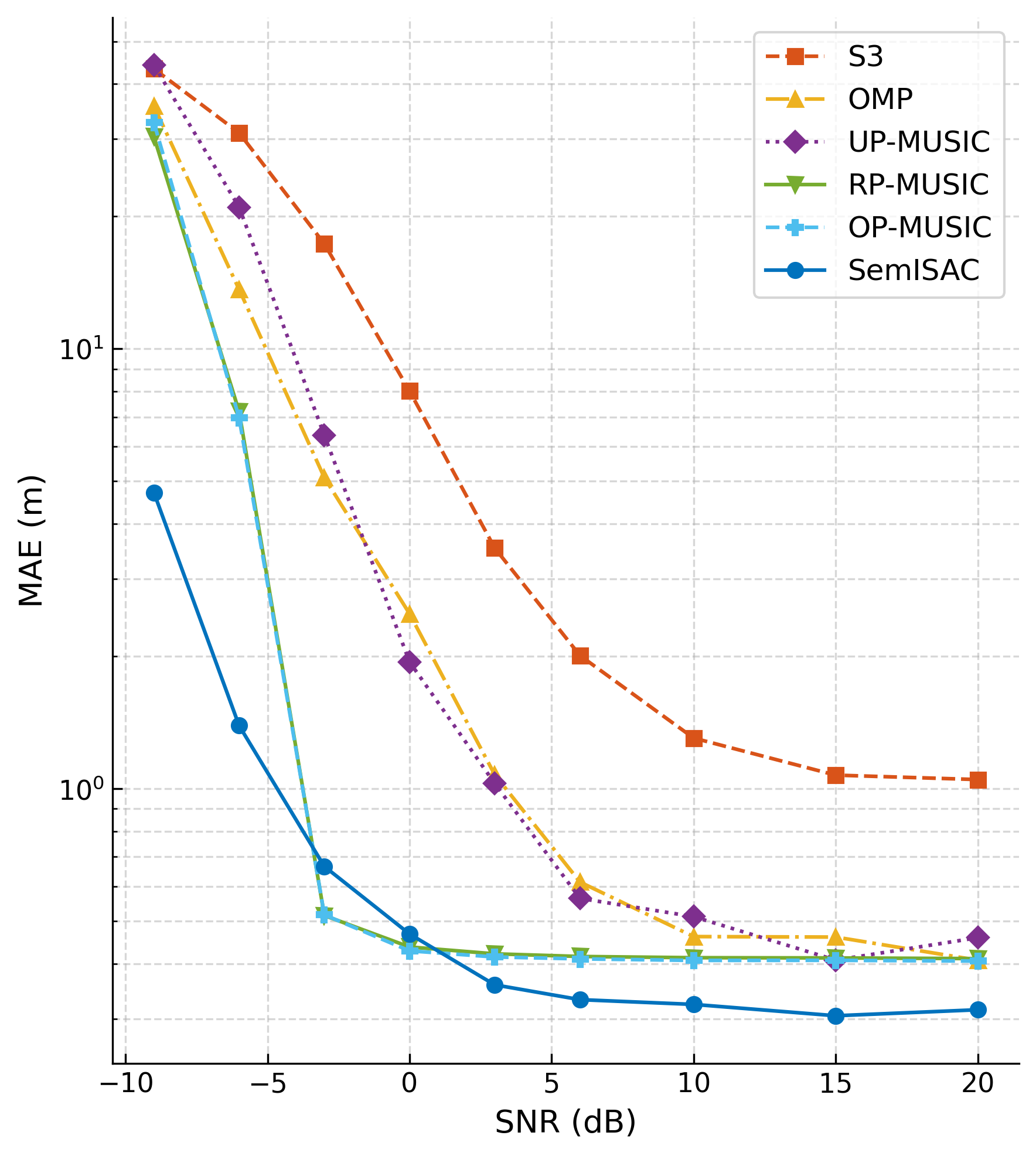}
        \caption{Sensing ranging MAE}
        \label{fig:graph_c}
    \end{subfigure}
    
    \caption{Performance of the sensing module of the SemISAC framework compared with various classification and ranging baselines.}
    \label{fig:combined_graphs}
\end{figure*}

\subsection{Adaptive Pilot Selection}
In an OFDM grid, each cell is used as a pilot when its selection score is above a threshold. These scores remain fixed across images; however, the threshold varies. A small network $\eta_\varphi(\gamma)$ considers the current signal-to-noise ratio $\gamma$ and determines the threshold, so a cell becomes a pilot only
when
\begin{equation}
    \mathbf{P}_{n,m} =
    \begin{cases}
        1, & g_{n,m} > \eta_\varphi(\gamma),\\
        0, & \text{otherwise.}
    \end{cases}
    \label{eq:gate}
\end{equation}
When the channel is
poor, the threshold is lowered by the neural network, so more cells are selected as pilot cells and the
target is measured more reliably. In contrast, when the channel is good the threshold rises and fewer cells are selected as pilots.
Because this is a hard yes-or-no choice, it does not allow any gradients to flow during backpropagation in neural networks. To
overcome this, we adopt a straight-through estimator strategy. During the backward pass of training this binary decision is treated as a sigmoid of the difference between the
score and the threshold.
\begin{equation}
        \tilde{\mathbf{P}}_{n,m}
    = \sigma\!\left( \frac{g_{n,m} - \eta_\varphi(\gamma)}{\omega} \right),
    \qquad
    \sigma(u) = \frac{1}{1 + e^{-u}},
    \label{eq:soft_gate}
\end{equation}
This smooth function provides a non-zero gradient,
allowing the neural network to learn. In this way, the model can be trained using
continuous gradients while the actual transmitted signal still uses exact
a binary pilot selection strategy.

\subsection{Training Objective}
The communication and sensing pipelines are jointly optimized with respect to three quantities: communication-task accuracy, sensing-task accuracy, and the number of pilots used. The communication quality is measured by the pixel-wise cross-entropy
$\mathcal{L}_{\text{comm}}$ between the predicted and the true road segmentation. The sensing quality is the classification cross-entropy
$\mathcal{L}_{\text{cls}}$ for target recognition; range estimation uses a normalized squared
error
\begin{equation}
    \mathcal{L}_{\text{rng}} = \Big( \frac{\hat{R}-R}{\sigma_R} \Big)^2 ,
    \qquad \sigma_R = \frac{R_{\max}}{\sqrt{12}},
\end{equation}
where $\sigma_R$ is the standard deviation of a range distributed uniformly over
$[0,R_{\max}]$. Since the communication and sensing losses differ in scale and behavior, combining them under a single loss with their original scale values would result in the improvement of only one pipeline. Therefore, a parameter $s_i$ learned by the neural network is used to adjust the weight according to the task. With this, the two losses of a pipeline are
combined as $\sum_i \big( \tfrac{1}{2} e^{-s_i}\mathcal{L}_i + \tfrac{1}{2}s_i
\big)$. 
Finally, to stop the network from using more pilots than it needs, we add a cost proportional to the fraction of pilot cells, $N_p/N_c$, so
that extra pilots are only kept when they are worth it. Combining these terms, the recognition and ranging pipelines minimize

\begin{align}
    \mathcal{L}_{\text{class}} &=
        \sum_{i\in\{\text{comm},\text{cls}\}}
        \Big( \tfrac{1}{2} e^{-s_i}\mathcal{L}_i + \tfrac{1}{2} s_i \Big)
        + \gamma_{\text{cls}}\, \frac{N_p}{N_c},\\
    \mathcal{L}_{\text{range}} &=
        \sum_{i\in\{\text{comm},\text{rng}\}}
        \Big( \tfrac{1}{2} e^{-s_i}\mathcal{L}_i + \tfrac{1}{2} s_i \Big)
        + \gamma_{\text{rng}}\, \frac{N_p}{N_c}.
\end{align}
The complete system is trained end-to-end, and the two communication and sensing pipelines are optimized in the same loop, each with its own Adam optimizer and gradient-norm clipping. In each mini-batch of training images, the SNR is drawn uniformly and a new channel is generated, so that the trained model learns to adjust its
pilot count across diverse operating conditions rather than at one fixed condition. Throughout the training, the power-allocation coefficient $\rho$ is kept constant and its
effect is analyzed separately by changing it at evaluation time.

\begin{figure*}[t]
    \centering
    
    \begin{subfigure}{1\columnwidth}
        \centering
        \includegraphics[width=0.8\linewidth]{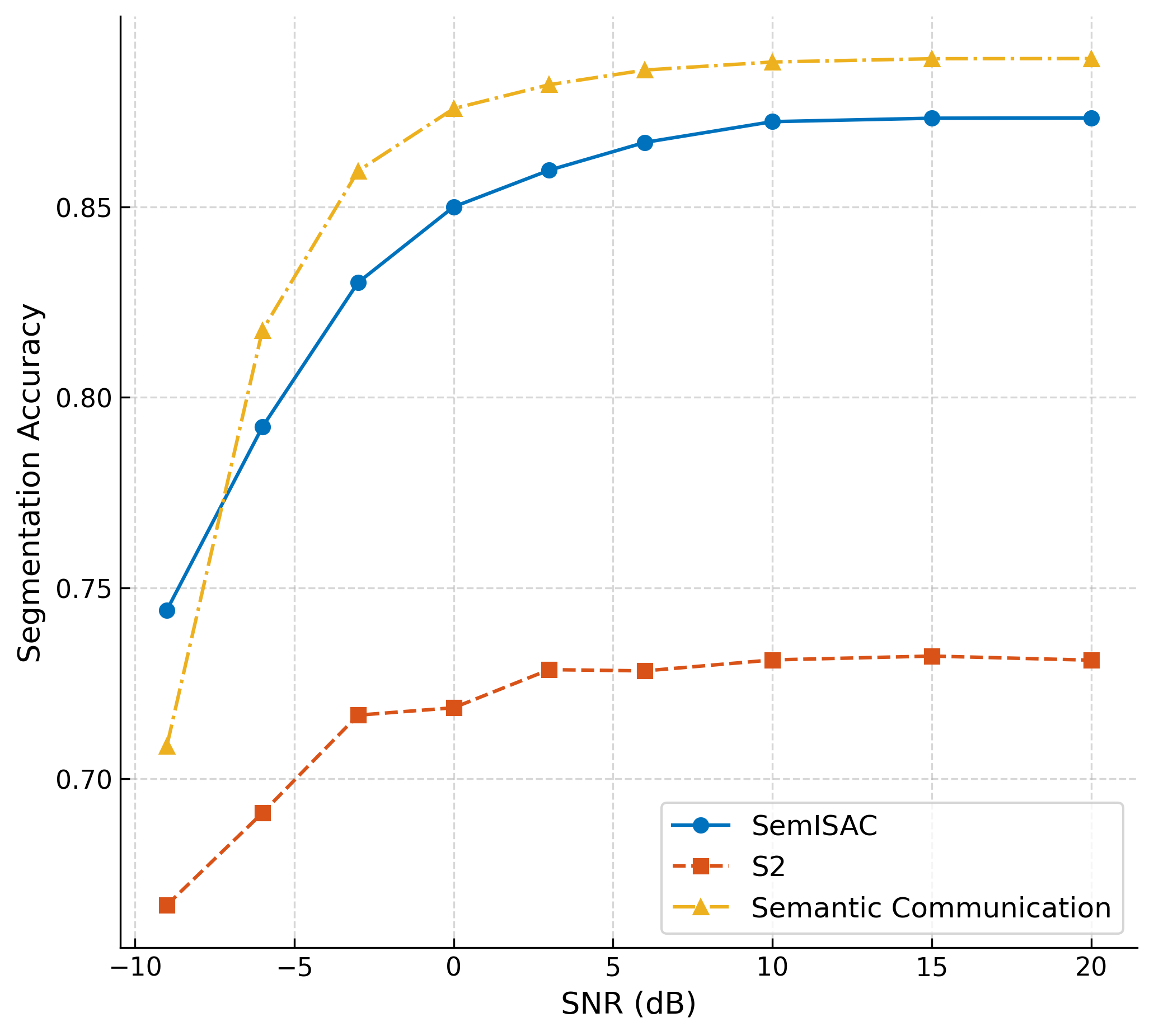}
        \caption{Performance comparison of SemISAC with baselines.}
        \label{fig:seg1}
    \end{subfigure}
    \begin{subfigure}{1\columnwidth}
        \centering
        \includegraphics[width=0.8\linewidth]{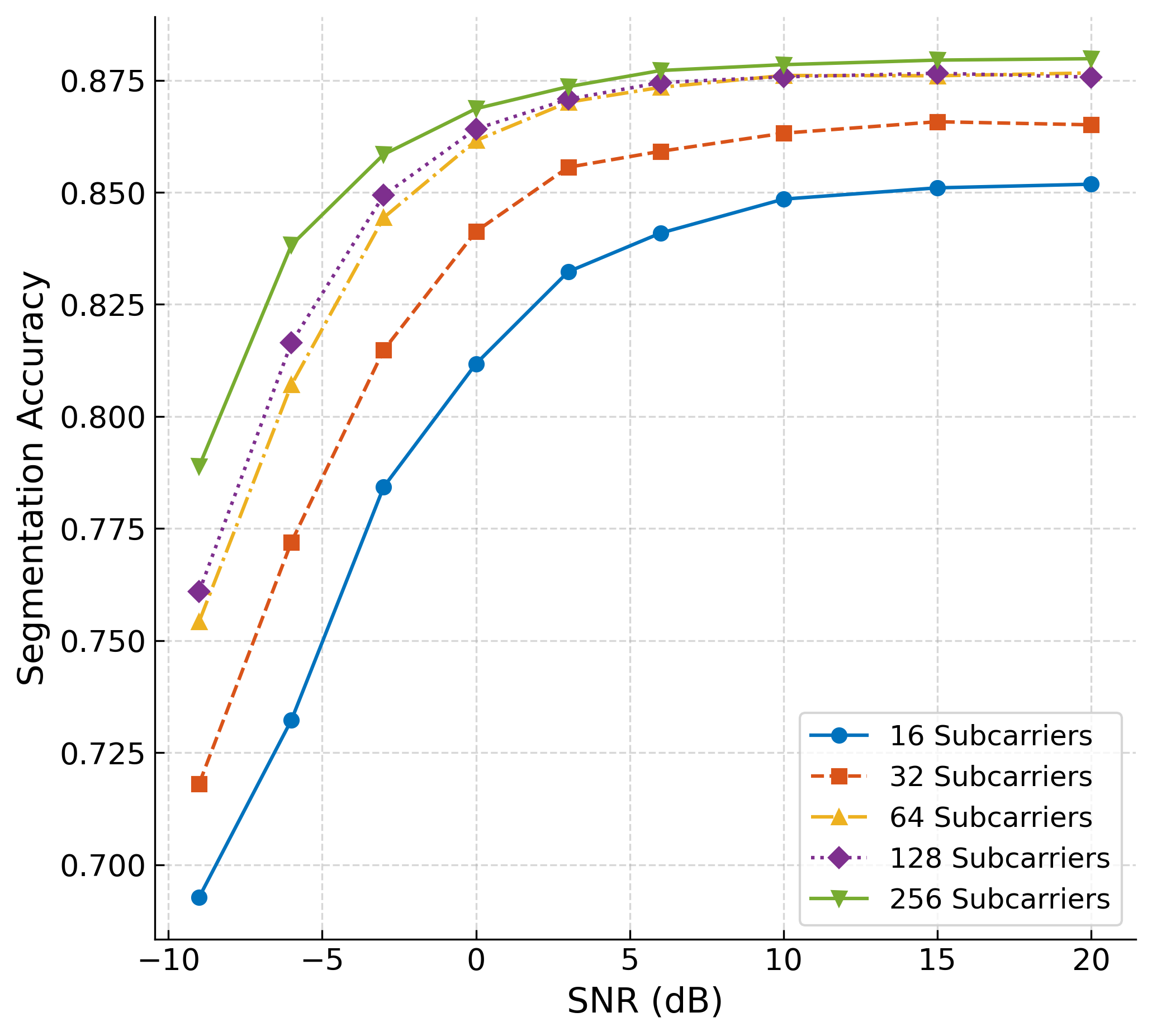}
        \caption{Effect of changing the number of subcarriers.}
        \label{fig:seg2}
    \end{subfigure}
    \caption{Evaluation of semantic communication in terms of segmentation accuracy.}
\end{figure*}

\section{Simulation Setup \& Results}
\label{results}
We evaluate the proposed approach integrating SemSens and SemCom using different metrics, different signal-to-noise ratios (SNRs), and a range of channels. The time--frequency grid, channel, and model parameters used
in the simulations are listed in
Table~\ref{tab:params}.

Performance is measured with different metrics for the communication and sensing branches. For the communication
task, the recovered segmentation is measured by pixel accuracy, mean
intersection-over-union (mIoU), and macro-averaged F1 score over the $O_s$
road classes.
For the sensing task, target recognition is measured by
classification accuracy and macro-F1, while range estimation is measured by the
mean absolute error (MAE) in meters. The three target classes---pedestrian, car, and drone---and their bulk-velocity
and micro-Doppler parameters used in the simulation are summarized in Table~\ref{tab:classes}.
We compare our proposed approach with several baseline configurations that use different combinations of conventional and semantic modules. To ensure a fair comparison, all baseline methods are reproduced in the same simulation environment and evaluated using identical channel conditions, waveform parameters, target distributions, and SNR levels.
\begin{table}[t]
\centering
\caption{Target objects.}
\label{tab:classes}
\begin{tabular}{lcc}
\hline
Class & Bulk velocity $v$ (m/s) & Micro-Doppler $f_\mu$ (Hz) \\
\hline
Pedestrian & $0$--$2$   & $1$--$3$    \\
Car        & $10$--$20$ & $0$ (rigid) \\
Drone      & $5$--$10$  & $80$--$120$ \\
\hline
\end{tabular}
\end{table}

We consider multiple scenarios with different semantic
capabilities. We first build a basic ISAC system that serves as the base structure of the framework; the system is then modified by adding semantic modules. \emph{Scenario~2}
(S2) replaces the conventional communication module of ISAC with SemCom while 
sensing remains conventional; this scenario is built on the basis of~\cite{semISAC3}. 
\emph{Scenario~3} (S3) instead replaces conventional sensing with a SemSens module reproduced from the study in~\cite{semsensmain}.
Finally, \emph{Scenario~4} (S4) is our proposed system, which
integrates \emph{both} a SemCom module and a SemSens
module. This setup lets us analyze the
benefit contributed by each semantic module individually---S2 for SemCom and
S3 for SemSens---and the importance of combining them into a single unified design.

\subsection{Semantic Communication Based ISAC (S2)}
S2 adopts the SemCom design of~\cite{semISAC3}, in which a deep joint source--channel coding (DJSC) \cite{djsc} network jointly performs image compression and channel coding. The semantic symbols are
quantized into bits and then mapped to a conventional
modulation scheme, such as PSK or QAM. The symbols are also assigned different
quantization levels, modulation orders, and transmit powers according to their
semantic importance. At the receiver, the DJSC decoder maps the received symbols to the pixel-wise segmentation of the scene, without reconstructing the original image.

\subsection{Semantic Sensing (S3)}
\label{subsec:s3} S3 uses the sensing module design from~\cite{semsensmain}. It uses a learned pilot-selection mechanism with
task-specific sensing decoders. It also uses a dictionary-based ranging head trained end-to-end for the sensing tasks that estimates the delay of the target. To evaluate the sensing task thoroughly, we reproduce the sensing baselines used in this paper. 
For \emph{target recognition}:
\begin{itemize}
\item \textbf{Full Channel Reconstruction (Perfect CSI):} an ideal reference
that uses the same backbone neural network as the proposed method but is provided with perfect channel state information during decoding.
\item \textbf{DL-MSE:} a two-stage scheme that first reconstructs the
channel by minimizing the mean-squared error and then classifies the target from the signal reconstruction. 
\item \textbf{Uniform SVM \& Optimized SVM:} classical radar-processing pipelines in which handcrafted statistical features of the Doppler spectrum (mean, standard deviation, skewness, and kurtosis) are first extracted and then classified by an SVM model. Uniform SVM uses evenly placed pilots across the grid, and optimized SVM uses a pilot placement optimized for the Doppler pattern.
\end{itemize}
\begin{figure}[t]
    \centering

    \begin{subfigure}[t]{0.25\linewidth}
        \centering
        \includegraphics[width=\linewidth]{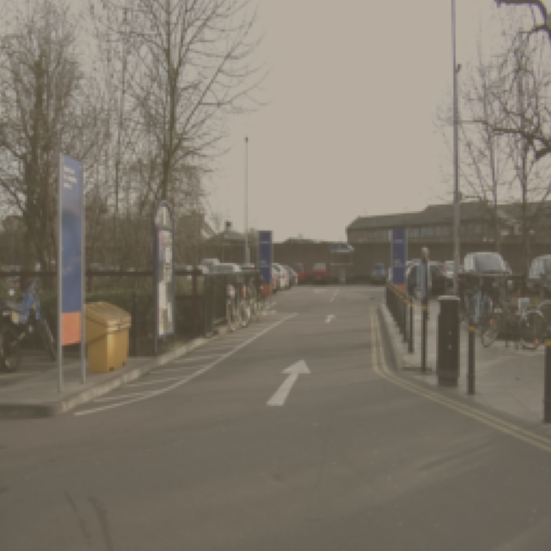}
        \caption{Image}
        \label{fig:road}
    \end{subfigure}
    \hfill
    \begin{subfigure}[t]{0.25\linewidth}
        \centering
        \includegraphics[width=\linewidth]{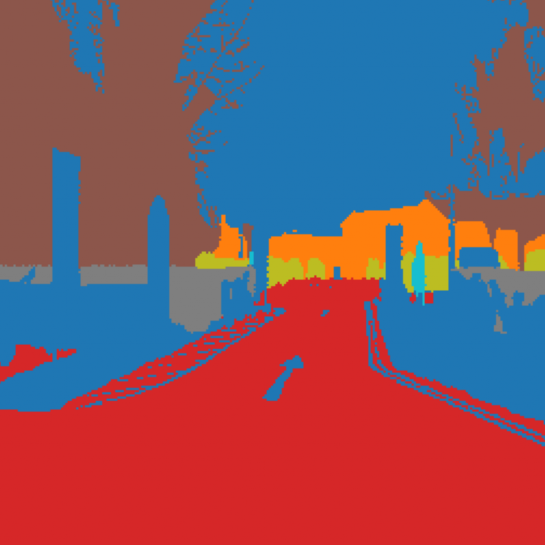}
        \caption{Ground truth segmentation}
        \label{fig:gt}
    \end{subfigure}
    \hfill
    \begin{subfigure}[t]{0.25\linewidth}
        \centering
        \includegraphics[width=\linewidth]{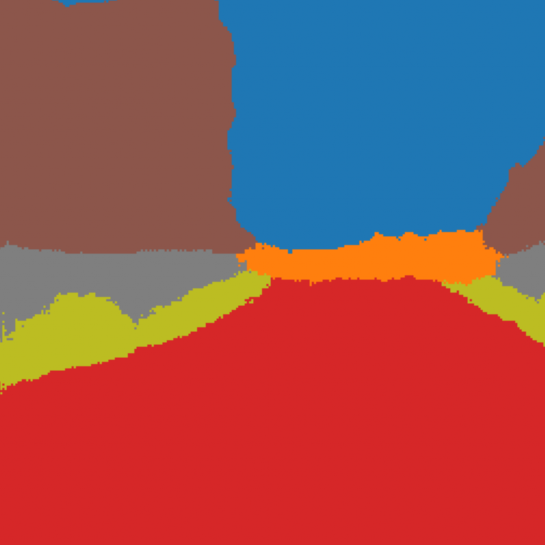}
        \caption{SemISAC}
        \label{fig:roadsemisac}
    \end{subfigure}

    \vspace{0.5em}

    \begin{subfigure}[t]{0.25\linewidth}
        \centering
        \includegraphics[width=\linewidth]{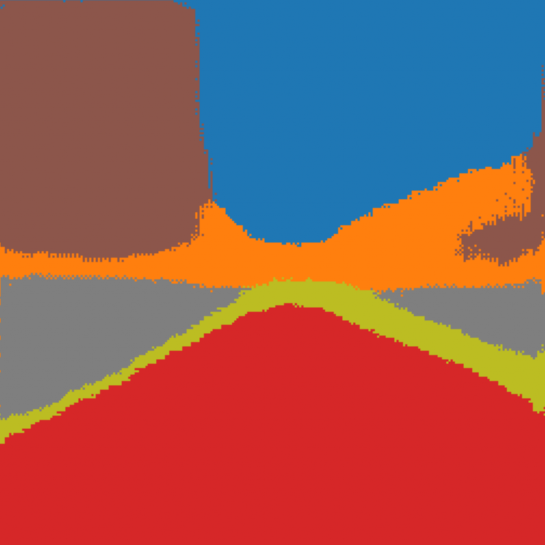}
        \caption{Simple semantic communication}
        \label{fig:roadsemcom}
    \end{subfigure}
    \hspace{20pt}
    \begin{subfigure}[t]{0.25\linewidth}
        \centering
        \includegraphics[width=\linewidth]{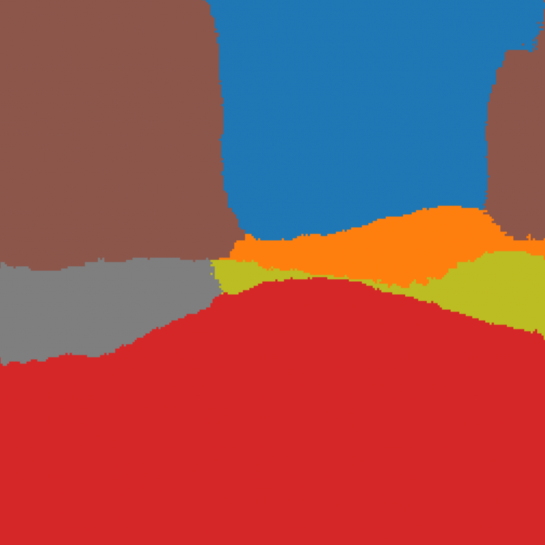}
        \caption{S2}
        \label{fig:roads2}
    \end{subfigure}

    \caption{An instance from the dataset, its ground truth segmentation, and the performance of all approaches.}
    \label{fig:main_figure}
\end{figure}
For \emph{range estimation}:
\begin{itemize}
\item \textbf{OMP:} orthogonal matching pursuit, which uses a list of possible
delays (each delay corresponds to a distance). It goes through that list and picks the delay that best matches the received echo, thereby estimating the target distance.
\item \textbf{UP-MUSIC}, \textbf{RP-MUSIC}, and \textbf{OP-MUSIC:} three different versions of the multiple signal classification (MUSIC) method, which split the received pilots into a target
part and a noise part, scan every possible delay, and retain the delays that match the target part to estimate the target distance.
\item \textbf{Blind:} a reference baseline that uses no channel state information and predicts the
mean range blindly.
\end{itemize}

\subsection{Performance Evaluation}
Figure \ref{fig:graph_a} illustrates the performance of our proposed methodology in terms of the sensing accuracy of object classification. In addition to the raw accuracy score, we also compare the F1 score in Figure \ref{fig:graph_b} for classification to provide a more comprehensive evaluation, as F1 considers both precision and recall. We can observe that the SemSens module of the proposed SemISAC methodology outperforms nearly all other methods across the board, including the simple SemSens methodology and DL-based channel reconstruction. The improved performance stems from the adaptive pilot design, which allows the proposed methodology to identify and favor pilot positions that provide the most informative observations for downstream tasks at test time, allowing the pilot resources to be allocated according to their sensing utility instead of under a fixed scheme, such as a uniform or heuristic allocation. The observed performance improvement over S3 and the DL-based method suggests that adaptive pilot placement is able to extract task-relevant information more effectively than fixed pilot configurations. Additionally, the conventional feature-based methods involving handcrafted features show the most limited performance. Compared to the adaptive learning methods, this underperformance can be attributed to the limitations of the handcrafted features, as they lack the ability to identify and utilize more informative features. The enhanced performance of the proposed methodology and the general trends are visible in both the accuracy and F1 graphs, further substantiating the results.

The sensing performance of the proposed methodology is also evaluated in terms of the mean absolute error (MAE) achieved in order to gauge the performance of the proposed method in continuous parameter estimation scenarios. As shown in Figure \ref{fig:graph_c}, RP-MUSIC and OP-MUSIC exhibit the least error among all the MUSIC variants, whereas S3 consistently shows the largest error. The proposed method performs optimally in this metric as well, recording an MAE of less than 10~m even at an SNR of $-9$~dB. The performance gap between adaptive and non-adaptive pilot designs is much more evident in this experiment, which can be attributed to the higher sensitivity of continuous parameter estimation to small measurement errors. While classification can remain correct despite moderate disturbances in the received signal, range estimation is much more sensitive to the accuracy of the estimated signal parameters. As a result, the improved measurement quality obtained by adaptive pilots is reflected more directly in the reduced MAE, making the advantage more pronounced.
\begin{figure}[t]
    \centering
    \includegraphics[width=0.8\linewidth]{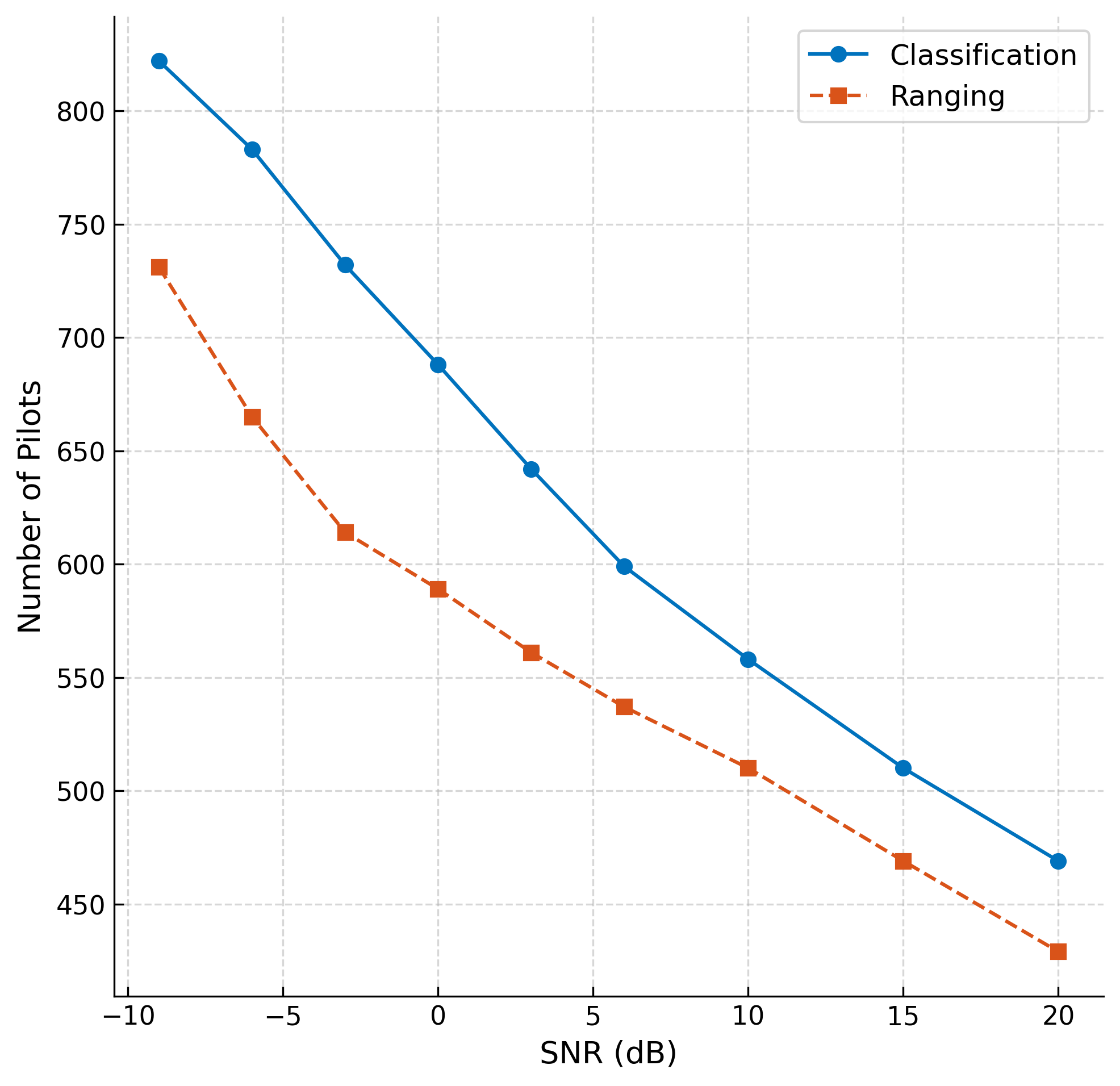}
    \caption{Number of pilots selected at different SNRs.}
    \label{fig:pilots_vs_snr}
\end{figure}

\begin{figure*}[t]
    \centering
    \begin{subfigure}{0.66\columnwidth}
        \centering
        \includegraphics[width=1\linewidth]{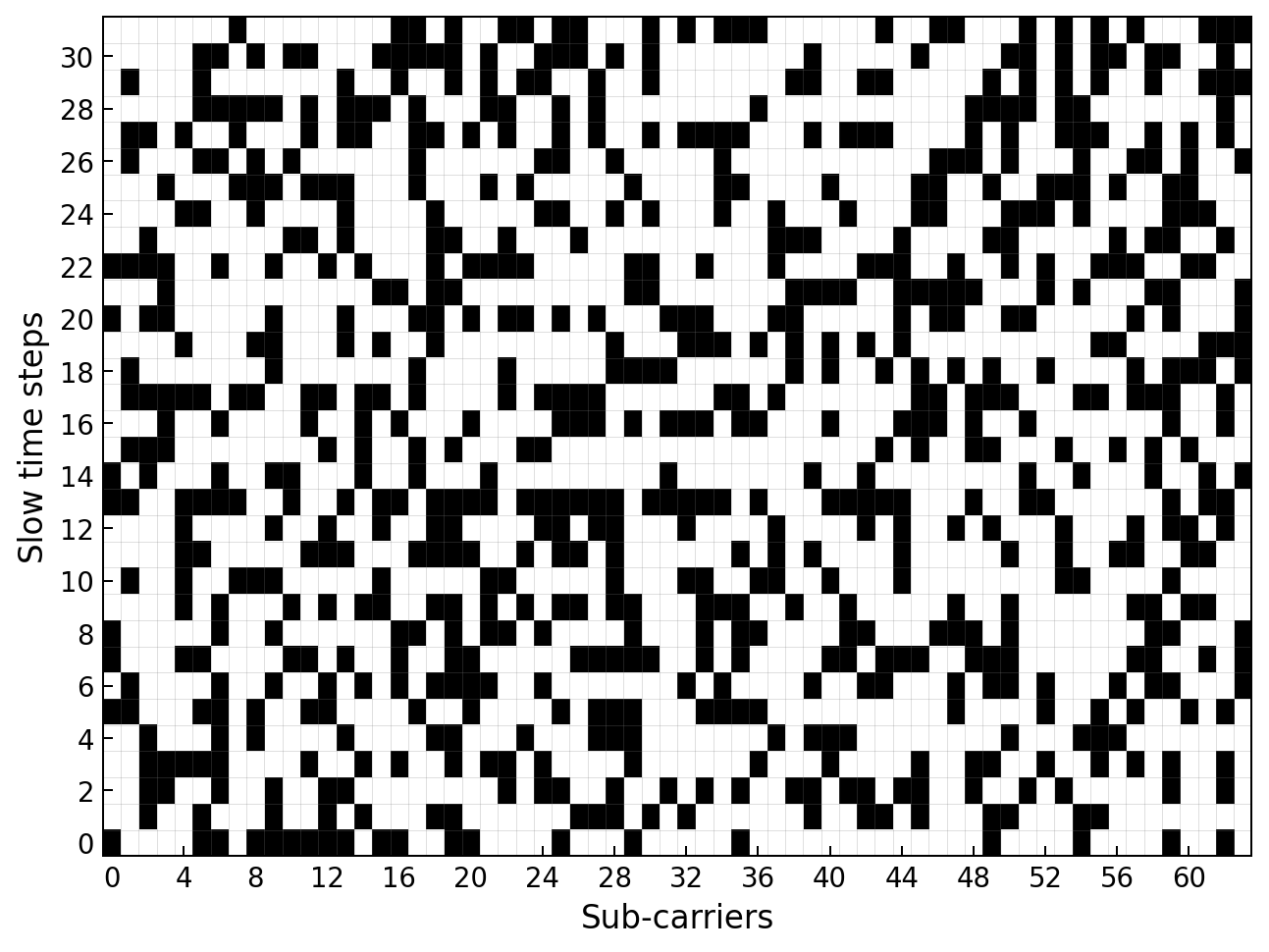}
        \caption{Ranging -9 dB}
        \label{fig:srm-9}
    \end{subfigure}
    \begin{subfigure}{0.66\columnwidth}
        \centering
        \includegraphics[width=1\linewidth]{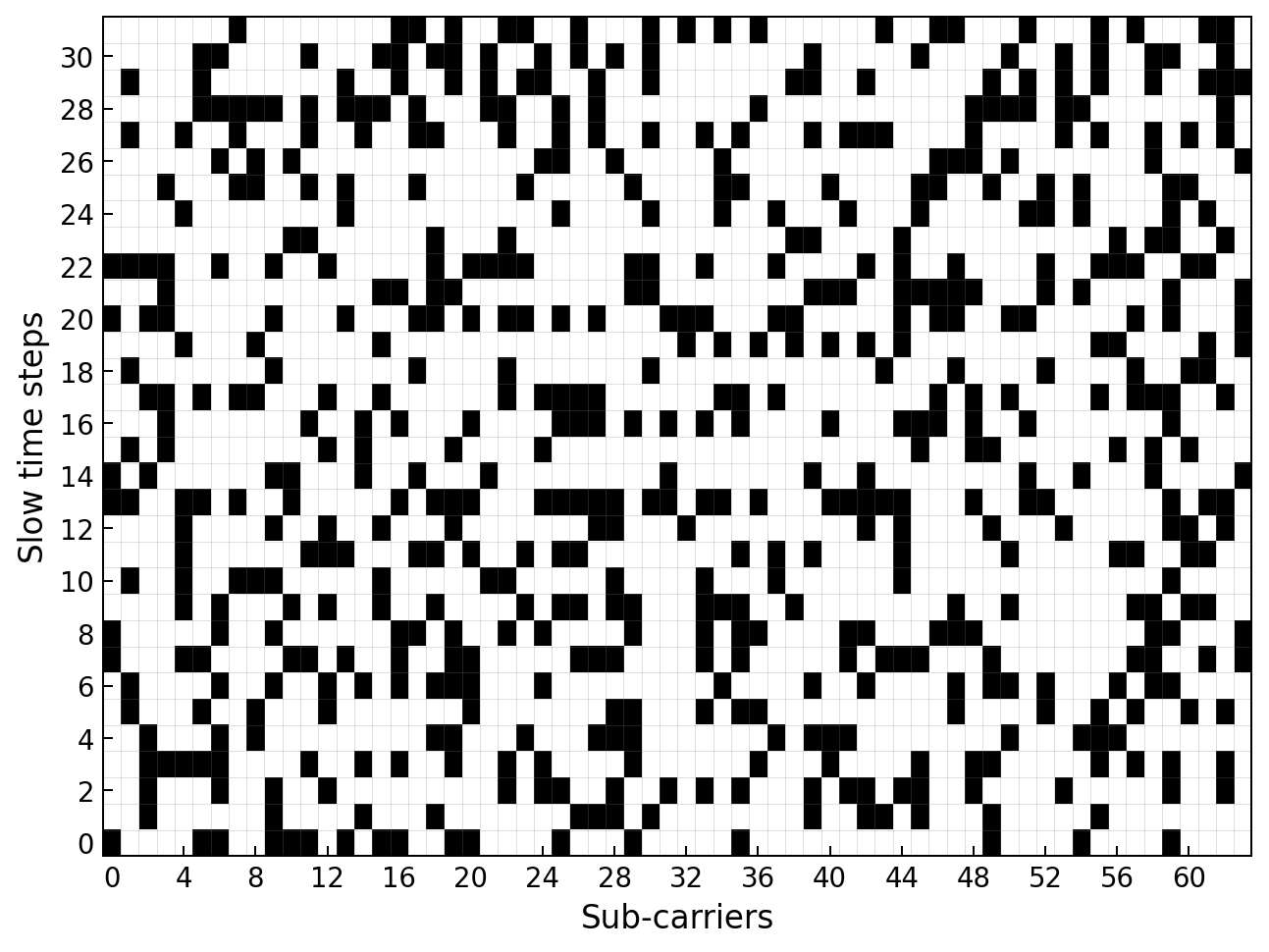}
        \caption{Ranging 0 dB}
        \label{fig:srm0}
    \end{subfigure}
    \begin{subfigure}{0.66\columnwidth}
        \centering
        \includegraphics[width=1\linewidth]{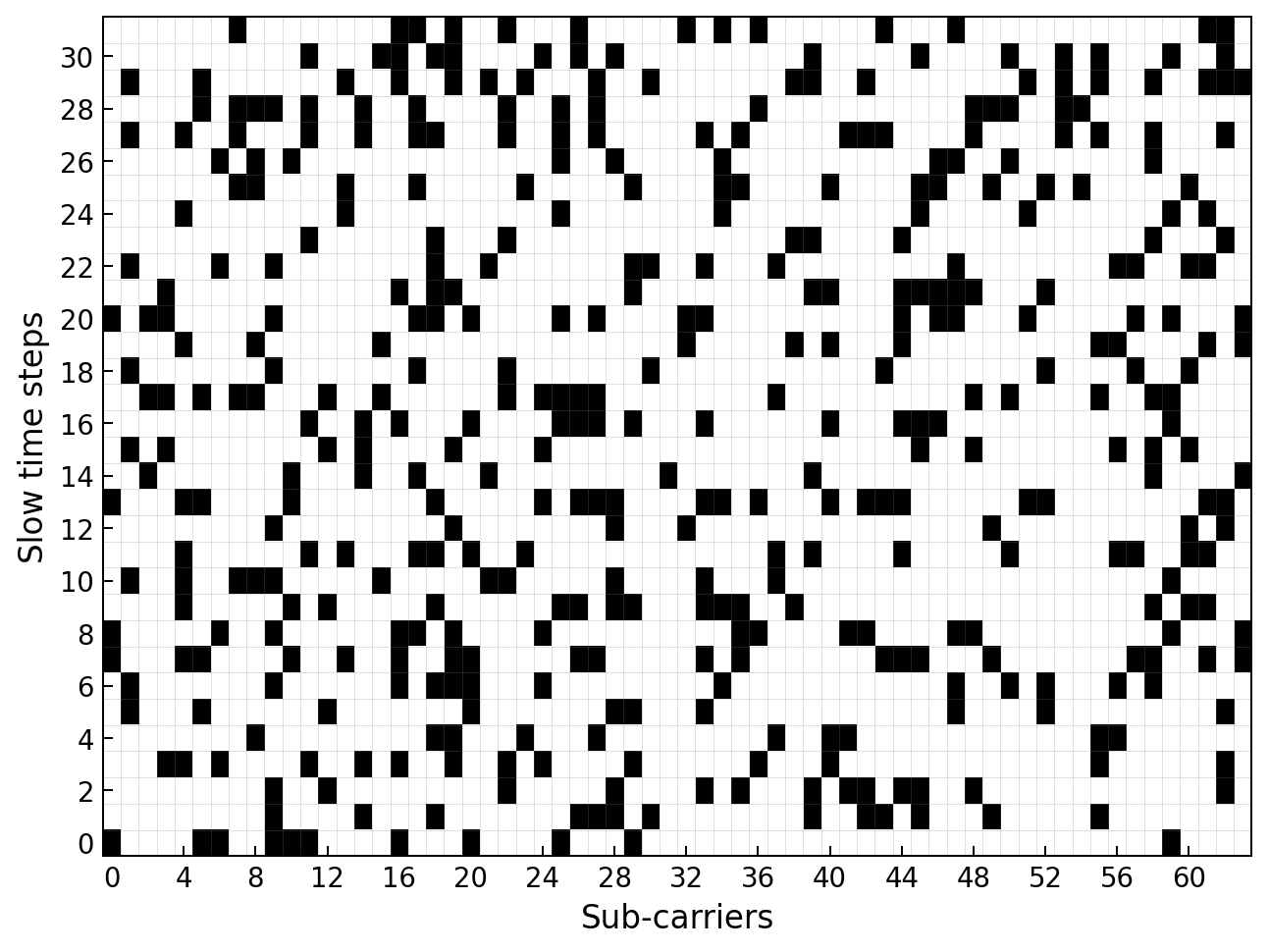}
        \caption{Ranging 20 dB}
        \label{fig:srm20}
    \end{subfigure}
    \begin{subfigure}{0.66\columnwidth}
        \centering
        \includegraphics[width=1\linewidth]{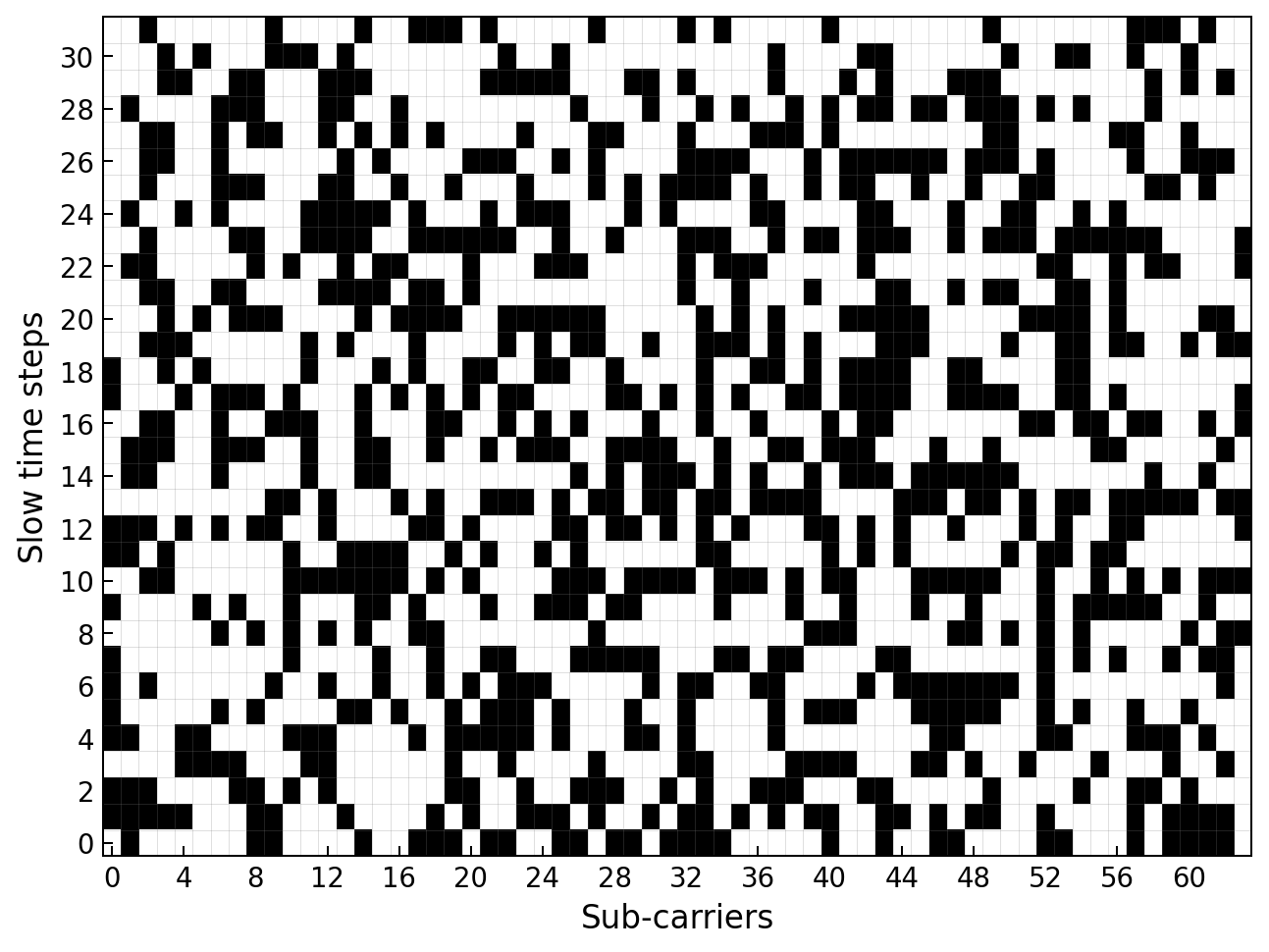}
        \caption{Classification -9 dB}
        \label{fig:scm-9}
    \end{subfigure}
    \begin{subfigure}{0.66\columnwidth}
        \centering
        \includegraphics[width=1\linewidth]{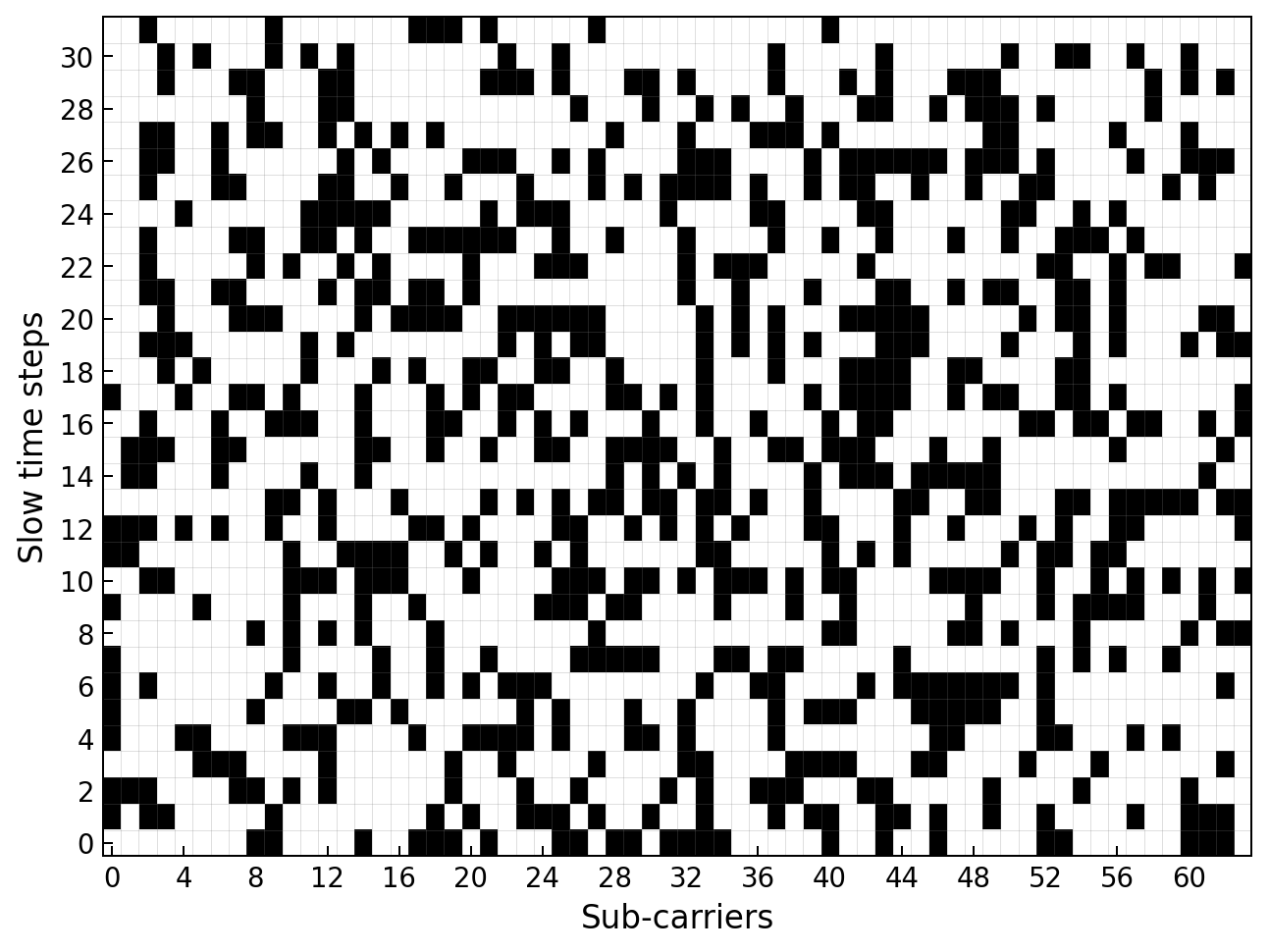}
        \caption{Classification 0 dB}
        \label{fig:scm0}
    \end{subfigure}
    \begin{subfigure}{0.66\columnwidth}
        \centering
        \includegraphics[width=1\linewidth]{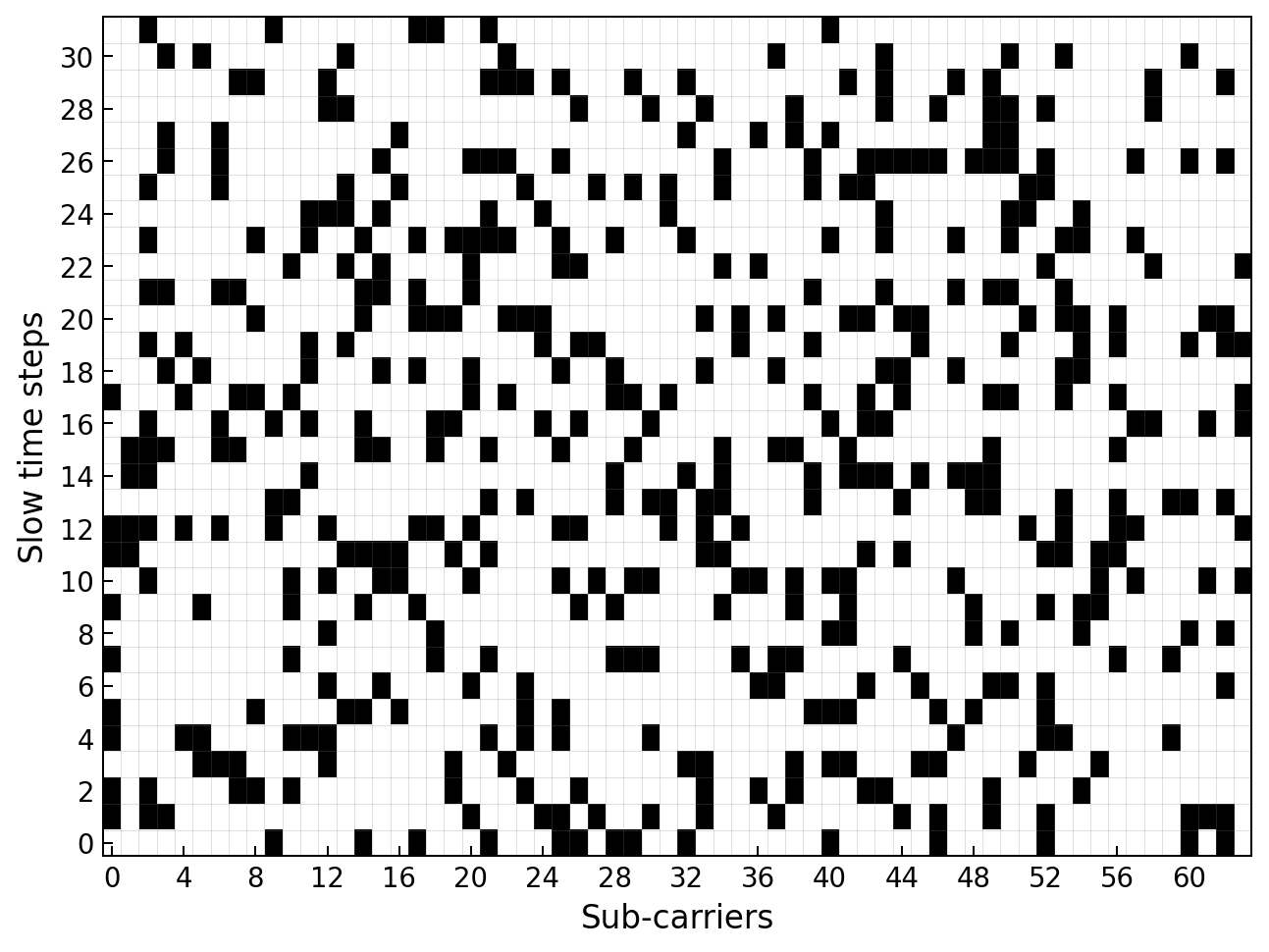}
        \caption{Classification 20 dB}
        \label{fig:scm20}
    \end{subfigure}
    
    \caption{The adaptive pilot placement of the proposed SemISAC methodology for both ranging and classification sensing scenarios at different SNRs.}
    \label{fig:masks}
\end{figure*}

Figure \ref{fig:seg1} presents the performance of the SemCom module of our proposed SemISAC methodology compared to S2 and simple SemCom. The simple SemCom method serves as a communication-only reference representing the performance achievable in the absence of sensing-related constraints. It provides a practical ceiling for comparison, as it isolates the communication task from the additional requirements imposed by sensing in an ISAC system. As shown by the results, the SemISAC approach achieves performance comparable to simple SemCom despite simultaneously performing sensing tasks. The improvement over S2 is driven by the Swin Transformer architecture used in the backbone of our proposed approach. At extremely low SNRs, the proposed SemISAC approach even outperforms simple SemCom due to its adaptive pilot design, which allows it to adjust its waveform based on channel and environmental conditions. A practical example of the image data, ground truth, and performance of all the approaches is provided in Figure \ref{fig:main_figure}, which shows that SemISAC and simple SemCom maintain high segmentation accuracy, whereas S2 introduces significant distortions in the task results.

Although the simple SemCom approach also uses the Swin architecture, the performance gains of simple SemCom over SemISAC are not as pronounced as the gain of SemISAC over S2. To further investigate this behavior, we evaluate the proposed approach with different numbers of subcarriers, as shown in Figure \ref{fig:seg2}. The results demonstrate diminishing returns after increasing the number of subcarriers beyond 64. Going from 16 to 32 and from 32 to 64 subcarriers leads to a percentage increase of 3.64\% and 2.41\%, respectively, whereas going beyond 64 subcarriers leads to a percentage increase of less than 1\%. This can be linked to the semantic expressive power of the resource elements. With 64 subcarriers, the resource elements are sufficient to encode the maximum amount of semantic information required for successful decoding at the receiver. Therefore, increasing the number of subcarriers beyond 64 does not result in a significant improvement in performance. This explains the minor gap shown by simple SemCom in Figure \ref{fig:seg1}, as both approaches are evaluated with 64 subcarriers, and accommodating pilot signals for sensing in SemISAC does not cause a significant drop in the expressive power of the SemISAC encoder.

Figure \ref{fig:masks} presents the adaptive pilot placement of the proposed SemISAC framework for both classification and ranging tasks at different signal strengths. Within the given $N \times M$ grid, the black cells represent the pilot REs and the white cells represent the data-bearing REs. It can be observed that as the SNR improves, the density of pilots decreases under favorable channel conditions. This behavior is shown in Figures \ref{fig:srm-9} through \ref{fig:srm20} for ranging and \ref{fig:scm-9} through \ref{fig:scm20} for classification and reflects the adaptive nature of the proposed strategy. When the received signal quality is high, reliable sensing information can be obtained from fewer pilots, allowing a larger fraction of the REs to be allocated to data transmission. Conversely, at lower SNRs, additional pilots are required to provide reliable observations. The observed variation and the decreasing trend in pilot density, as shown by both \ref{fig:masks} and \ref{fig:pilots_vs_snr}, demonstrate that the proposed methodology dynamically trades pilot density against signal quality instead of relying on a fixed pilot pattern. This observed behavior is consistent with the analytical characterization established in Proposition~\ref{thm:npilots}. Specifically, the proposition shows that under the considered distortion model, the optimal pilot count $N_p^\star(\gamma)$ is strictly decreasing with SNR and approaches zero in the high-SNR regime. The decreasing pilot density observed in the learned masks and in Figure~\ref{fig:pilots_vs_snr} therefore provides empirical evidence that the proposed SNR-conditioned pilot selection captures the qualitative trend predicted by the analytical model.

\section{Conclusion}
\label{conc}
In this study, we proposed SemISAC, a unified framework that performs both SemCom and SemSens using a single dual-function waveform. SemISAC uses a shared Swin Transformer encoder that converts the input road-scene image into semantic symbols placed on the OFDM grid. An adaptive pilot-selection network, trained with a straight-through estimator, determines how many pilots to use and where to place them based on the SNR. At the receiver, the communication decoder reconstructs the road-scene segmentation from the received data cells, while the channel state information is estimated from the pilot cells. The sensing decoders at the transmitting vehicle process the channel estimate obtained from the pilot cells and output the target class and range. Adaptive pilot selection is a key contributor to the performance gains: the selected pilot count decreases steadily as the SNR improves, allowing the system to use fewer pilots when the signal is strong. The evaluation results show that SemISAC achieves an accuracy close to that of the SemCom-only baseline and clearly outperforms the S2 baseline. It also outperforms the semantic and conventional sensing baselines, including S3, DL-MSE, SVM, OMP, and MUSIC, in classifying pedestrians, cars, and drones and in estimating their ranges.

%
\appendices
\section{Proof of Lemma~\ref{thm:rho}}\label{app:proof-rho}
\begin{proof}
Write $J(\rho)=\dfrac{wa}{\rho}+\dfrac{(1-w)b}{1-\rho}$. Then
$J''(\rho)=\dfrac{2wa}{\rho^{3}}+\dfrac{2(1-w)b}{(1-\rho)^{3}}>0$ on $(0,1)$, so $J$ is strictly convex. Since $J'(\rho)\to-\infty$ as $\rho\to0^+$ and $J'(\rho)\to+\infty$ as $\rho\to1^-$, a unique interior stationary point exists. Setting $J'(\rho)=-\dfrac{wa}{\rho^{2}}+\dfrac{(1-w)b}{(1-\rho)^{2}}=0$ gives $\dfrac{\rho}{1-\rho}=\sqrt{\dfrac{wa}{(1-w)b}}$, and solving yields \eqref{eq:rhostar}. Writing $\rho^\star=\bigl(1+\sqrt{b/a}\,\sqrt{(1-w)/w}\bigr)^{-1}$, the factor $\sqrt{(1-w)/w}$ is strictly decreasing in $w$, so $\rho^\star$ is strictly increasing in $w$. The stated limiting values follow directly.

This concludes the proof.
\end{proof}

\section{Proof of Proposition~\ref{thm:npilots}}\label{app:proof-npilots}
\begin{proof}
The optimization has the same strictly convex form as in Lemma~\ref{thm:rho}. Specifically, defining $A=\dfrac{wa}{\gamma}$, $C=\dfrac{(1-w)c}{\log(1+\gamma)}$ gives $J(N_p)=\dfrac{A}{N_p}+\dfrac{C}{N_c-N_p}$. Hence, by the same allocation argument as in Lemma~\ref{thm:rho}, the unique minimizer is $N_p^\star=N_c\sqrt{A}/(\sqrt{A}+\sqrt{C})$, which yields \eqref{eq:npstar}. It remains to establish the SNR dependence. Writing $N_p^\star=N_c(1+\sqrt{C/A})^{-1}$, we have $\dfrac{C}{A}=\dfrac{(1-w)c}{wa}\,g(\gamma)$, where $g(\gamma)=\dfrac{\gamma}{\log(1+\gamma)}$. Then $g'(\gamma)=[\log(1+\gamma)-\tfrac{\gamma}{1+\gamma}]/\log^2(1+\gamma)$. Let $h(\gamma)=\log(1+\gamma)-\tfrac{\gamma}{1+\gamma}$. Since $h(0)=0$ and $h'(\gamma)=\tfrac{\gamma}{(1+\gamma)^2}>0$ for $\gamma > 0$, we have $h(\gamma)>0$, and therefore $g'(\gamma)>0$. Thus $C/A$ is strictly increasing in $\gamma$, implying that $N_p^\star(\gamma)$ is strictly decreasing. Finally, $\dfrac{\gamma}{\log(1+\gamma)} \to \infty$ as $\gamma \to \infty$, so $C/A\to\infty$ and, consequently, $N_p^\star(\gamma) \to 0$.
\end{proof}




\ifCLASSOPTIONcaptionsoff
  \newpage
\fi

\end{document}